\documentclass[11pt]{article}
\RequirePackage[T1]{fontenc}
\RequirePackage[utf8]{inputenc}
\RequirePackage[english]{babel}
\RequirePackage{lmodern}
\RequirePackage{microtype}
\RequirePackage{csquotes}
\RequirePackage[a4paper,top=2.35cm,bottom=2.35cm,left=2.55cm,right=2.55cm,footskip=1.15cm]{geometry}
\RequirePackage{amsmath,amssymb,amsthm,mathtools}
\RequirePackage{booktabs,tabularx,longtable,array,multirow}
\newcolumntype{Y}{>{\raggedright\arraybackslash}X}
\newcolumntype{L}[1]{>{\raggedright\arraybackslash}p{#1}}
\RequirePackage{graphicx}
\RequirePackage{enumitem}
\RequirePackage{xcolor}
\definecolor{elpblack}{gray}{0.00}
\definecolor{elpgraydark}{gray}{0.25}
\definecolor{elpgraymid}{gray}{0.55}
\definecolor{elpgraylight}{gray}{0.96}
\RequirePackage[hyperfootnotes=false]{hyperref}
\RequirePackage[nameinlink,noabbrev]{cleveref}
\RequirePackage{titlesec}
\RequirePackage{fancyhdr}
\RequirePackage{setspace}
\RequirePackage{caption}
\RequirePackage{subcaption}
\RequirePackage{url}
\RequirePackage{float}
\RequirePackage{needspace}
\RequirePackage{tikz}
\usetikzlibrary{arrows.meta,positioning,calc,fit,backgrounds,decorations.pathreplacing,patterns}
\RequirePackage[
 backend=biber,style=authoryear,natbib=true,maxcitenames=2,maxbibnames=20,
 uniquename=false,doi=true,url=true,isbn=false,giveninits=true,dashed=false
]{biblatex}
\defbibheading{bibliography}[\refname]{\section*{#1}}
\hypersetup{hidelinks,pdfborder={0 0 0}}
\setlist{nosep,leftmargin=*}
\titleformat{\section}{\Needspace{7\baselineskip}\large\bfseries}{\thesection}{0.65em}{}
\titleformat{\subsection}{\Needspace{5\baselineskip}\normalsize\bfseries}{\thesubsection}{0.65em}{}
\titleformat{\subsubsection}{\normalsize\itshape}{\thesubsubsection}{0.65em}{}
\fancypagestyle{plain}{\fancyhf{}\fancyfoot[C]{\small\thepage}}
\newtheorem{theorem}{Theorem}[section]
\newtheorem{proposition}[theorem]{Proposition}

\theoremstyle{definition}
\newtheorem{definition}[theorem]{Definition}
\newtheorem{criterion}[theorem]{Criterion}

\theoremstyle{remark}

\newcommand{\E}{\mathbb{E}}

\newcommand{\ind}{\mathbf{1}}

\newcommand{\paperversion}[1]{\gdef\PaperVersion{#1}}
\newcommand{\paperseries}[1]{\gdef\PaperSeries{#1}}
\newcommand{\papersubtitle}[1]{\gdef\PaperSubtitle{#1}}
\newcommand{\paperaffiliation}[1]{\gdef\PaperAffiliation{#1}}
\newcommand{\papercontact}[1]{\gdef\PaperContact{#1}}

\AtBeginDocument{%
  \providecommand{\PaperVersion}{draft}%
  \providecommand{\PaperSeries}{ForesightFlow Research \textperiodcentered{} Event-Linked Perpetuals}%
  \providecommand{\PaperSubtitle}{}%
  \providecommand{\PaperAffiliation}{}%
  \providecommand{\PaperContact}{}%
  \providecommand{\PaperStatus}{}%
}
\newcommand{\authoritynote}[1]{\begin{center}\begin{minipage}{0.92\linewidth}\small\textit{#1}\end{minipage}\end{center}}
\tikzset{
 pbox/.style={draw=black,rounded corners=2pt,align=center,inner sep=5pt,font=\footnotesize,text width=23mm,minimum height=9mm,fill=white},
 pwide/.style={draw=black,rounded corners=2pt,align=center,inner sep=5pt,font=\footnotesize,text width=31mm,minimum height=10mm,fill=white},
 psoft/.style={draw=black,dashed,rounded corners=2pt,align=center,inner sep=5pt,font=\footnotesize,text width=26mm,minimum height=9mm,fill=white},
 parr/.style={-{Latex[length=2.0mm]},line width=0.65pt,draw=black},
 pdash/.style={-{Latex[length=2.0mm]},line width=0.65pt,dashed,draw=black},
 pbrace/.style={decorate,decoration={brace,amplitude=4pt},line width=0.55pt,draw=black},
 ffbox/.style={pbox},ffsmall/.style={pbox,font=\scriptsize,inner sep=4pt},ffsoft/.style={psoft},ffarrow/.style={parr},ffdashed/.style={pdash}
}
\makeatletter
\renewcommand{\maketitle}{%
  \begin{center}
    {\small\scshape \PaperSeries\par}
    \vspace{0em}
    {\LARGE\bfseries \@title\par}
    \ifx\PaperSubtitle\@empty\else\vspace{0em}{\normalsize \PaperSubtitle\par}\fi
    \vspace{0em}
    {\large \@author\par}
    \ifx\PaperAffiliation\@empty\else\vspace{0em}{\small \PaperAffiliation\par}\fi
    \ifx\PaperContact\@empty\else\vspace{0em}{\small \PaperContact\par}\fi
    \vspace{0em}{\small \@date\par}
  \end{center}
  \ifx\PaperStatus\@empty\else\authoritynote{\PaperStatus}\fi
  \vspace{0em}
}
\makeatother

\hypersetup{pdftitle={Two Models of Event Finality: Functional Alignment, Contestability, and Empirical Comparability on Polymarket and Kalshi},pdfauthor={Maksym Nechepurenko},pdfsubject={Event-Linked Perpetuals --- Kalshi Research Track: Cross-Venue Comparison},pdfkeywords={prediction markets, event contracts, finality, settlement, Polymarket, Kalshi, comparability},pdflang={en}}
\paperseries{ForesightFlow Research \textperiodcentered{} Event-Linked Perpetuals \textperiodcentered{} Paper 8}
\paperversion{r0.8.6}
\papersubtitle{Functional Alignment, Contestability, and Empirical Comparability on Polymarket and Kalshi}
\title{Two Models of Event Finality}
\author{Maksym Nechepurenko}
\paperaffiliation{Research Department of Devnull FZCO, Dubai, UAE}
\papercontact{\href{mailto:maksym@devnull.ae}{maksym@devnull.ae} \quad ORCID 0000-0002-9515-8841}
\date{September 2026 \textbar{} Version r0.8.6}
\begin{document}
\maketitle
\begingroup
\renewcommand{\thefootnote}{}
\footnotetext{\footnotesize\textit{ELP --- Kalshi Research Track: Cross-Venue Comparison.}}
\endgroup
\begin{abstract}
Event contracts can reach economic finality through different institutional paths. Polymarket evidence distinguishes oracle adjudication, adapter consumption, Conditional Tokens payout recording, technical redeemability, and optional holder redemption. Kalshi evidence distinguishes venue determination, public finalization, lifecycle messages, and exact REST settlement fields. Similar words such as \emph{resolved}, \emph{settled}, and \emph{finalized} therefore do not define comparable endpoints by themselves.

This paper develops a mechanism-aware comparison certificate whose coordinates are functional roles rather than common labels. A scalar cross-venue duration is admissible only when payoff semantics, endpoint functions, calendar support, censoring, observation grades, and dependence units are jointly aligned. The design separates exact paired, interval-qualified paired, and standardized unpaired estimands.

Both venue-native empirical layers are available. Published Polymarket evidence contains 108,638 exact-linked conditions, 99,283 protocol payout records, 92,158 observed redemptions of any amount, and 91,817 observed positive-payout redemptions. Kalshi evidence contains 152,694 ordinary markets reconstructed as at risk at the enrollment boundary and 7,611,594 exact MVE market objects. Exact public endpoints are observed for 71,657 ordinary markets and 7,357,576 MVE objects. Exact determination-to-endpoint pairs number 70,979 for ordinary markets and 126,806 for MVE.

These are Kalshi-native finality results, not Polymarket--Kalshi estimates. REST endpoint completeness can be high when exact lifecycle paths are sparse, and transport gaps can preserve two exact clocks while preventing claims about intermediate revisions. Current MVE timer fields do not identify historical version-consistent timer rules, and public finalization does not identify member cash. Cross-venue numerical comparison remains blocked until Paper~7.3 supplies blind first/stable decidability clocks and a semantic/calendar registry establishes genuinely comparable events and endpoints.
\end{abstract}

\noindent\textbf{Keywords:} prediction markets; event contracts; finality; settlement; Polymarket; Kalshi; contestability; matching; censoring.\\
\textbf{JEL:} G13, G14, G18.

\par\break
\section{Introduction}\label{sec:introduction}

A prediction market does not become final at one universal instant. A real-world event can occur before an official source publishes evidence. The evidence can become contractually sufficient before a venue records a result. A first result can remain contestable. A final mechanism result can precede the state that makes a claim redeemable or credits positions. A technically redeemable claim can remain unredeemed by its holder. Each interval belongs to a different actor, balance sheet, and failure mode.

This distinction becomes visible when two venues implement the same broad economic proposition through different architectures. Polymarket uses an oracle-mediated and smart-contract-mediated path. Kalshi uses venue-administered market states, rules, review, and settlement. Calling both terminal records ``settlement'' does not make them the same event. A comparison that subtracts whichever two timestamps happen to be available may rank labels rather than mechanisms.

The problem is not solved by adding more timestamps. A comparison also requires semantic equivalence. Two contracts can settle to the same realized answer while defining different payoff functions over states that did not occur: different thresholds, source precedence, correction policies, postponement rules, time zones, invalidity clauses, or scalar rounding. Ex post answer agreement is therefore weaker than ex ante contractual equivalence. Likewise, a venue-wide latency distribution can differ because of event composition even when the underlying mechanisms behave identically within every comparable stratum.

This paper asks four questions:
\begin{enumerate}
  \item Which venue-native states perform comparable economic functions, and which merely share names?
  \item Under what conditions does a cross-venue duration contrast estimate the same quantity?
  \item How should contracts be matched when titles, rules, sources, timing conventions, and payoff maps differ?
  \item Which claims survive incomplete clocks, non-overlapping evidence windows, differential censoring, or a small exact-match cohort?
\end{enumerate}

The paper's answer is a comparison certificate rather than a universal timestamp. Every candidate pair is characterized by a payoff-equivalence tier, an endpoint map, coordinate-specific observation grades, a calendar-support class, a censoring rule, and an inference cluster. A scalar duration enters the primary paired analysis only if all of these elements pass. Otherwise the observation remains useful in a lower-strength analysis: interval bounds, terminal-state comparison, standardized unpaired distributions, or descriptive institutional evidence.

\newpage
\subsection{Evidence inputs}\label{subsec:evidence-inputs}

The comparison combines published Polymarket mechanism evidence with distinct Kalshi population, lifecycle, and study-design inputs. Papers 5 Part I and II are public SSRN sources for Polymarket oracle adjudication, adapter terminality, protocol payout recording, and observed redemption \parencite{p5a2026,p5b2026}. Their denominators and claim boundaries remain unchanged.

The KMVE V1 dataset defines the Kalshi seven-day MVE population, parent hierarchy, structural signatures, current activity fields, and REST/WebSocket population crosswalk \parencite{kmve2026}. Paper 7.2 specifies how this structure determines the relevant dependence units, while Paper 7.4 specifies the prospective price-study design without providing a cross-venue price estimate \parencite{p72_2026,p74_2026}.

Kalshi lifecycle evidence supplies native cohorts, exact endpoint-source classes, and exact determination-to-endpoint pairs \parencite{p71_2026}. The ordinary cohort contains 152,694 markets reconstructed as at risk at the enrollment boundary. The MVE cohort contains 7,611,594 REST enrollment objects, of which 7,357,576 have an exact public endpoint by the follow-up boundary. Source-publication and contractual-decidability comparisons additionally require the separately specified Paper 7.3 record \parencite{p73_2026}.

The endpoint join shows why observation surface is itself a comparison coordinate. Most MVE exact endpoints are REST-only, while a smaller subset has an exact WebSocket determination path. Endpoint completeness and path completeness are therefore distinct. Historical-prefix reconstruction validates interim first/last pair populations and the equal-second class, but does not define the final cohorts or endpoint union.

Semantic equivalence and common-calendar support remain necessary conditions for a paired comparison. Topic similarity does not pair the prospective Kalshi population to the published Polymarket population, and no venue ranking is produced from native duration distributions.

\subsection{Contributions}\label{subsec:contributions}

The paper makes nine contributions.

\textbf{First}, it defines a functional finality vector that aligns economic roles without renaming venue-native states. Polymarket proposal is not called Kalshi determination; Oracle settlement is not called exchange finalization; protocol redemption is not called member cash settlement.

\textbf{Second}, it introduces a comparison certificate combining semantic match, functional endpoint alignment, observation grade, calendar support, censoring, and dependence unit. The certificate makes comparability an auditable data object rather than prose judgment.

\textbf{Third}, it strengthens the endpoint-alignment result. A duration contrast is admissible only when both endpoint functions and their evidence grades are compatible. Equal labels, equal realized answers, or equal timestamps are insufficient.

\textbf{Fourth}, it specifies payoff-preserving Tier A/B/C matching. Tier A requires equivalence of the contractual payoff functional across admissible states, not merely a common title or final answer.

\textbf{Fifth}, it treats incomplete and non-overlapping observation as part of the estimand. Exact, interval-observed, snapshot-only, proxy, unmeasured, conflicting, and right-censored coordinates remain distinct. Midpoints and convenient downstream substitutes are excluded from primary inference.

\textbf{Sixth}, it separates paired and standardized unpaired questions. A matched-event contrast asks how two mechanisms handled the same contractual event. A standardized contrast asks how mechanism-specific distributions differ under a common composition. These are not interchangeable estimands.

\textbf{Seventh}, it defines a multidimensional outcome space: latency, path complexity, contestability, revision, rule/source observability, governance concentration, venue/protocol terminality, cash realization, and capital lock. No total ranking is asserted without an explicit weighting function or Pareto dominance.

\textbf{Eighth}, it integrates the published Polymarket evidence layer while retaining the Kalshi and matching conditions. A small or empty Tier-A cohort is an empirical result about cross-venue contractual overlap, not a reason to relax semantic criteria.

\textbf{Ninth}, it treats population hierarchy as part of cross-venue identification. Market objects, event keys, collection identities, selected-leg primitives, structural signatures, and current activity fields are distinct statistical units; child-ticker multiplicity cannot be converted into independent finality observations.

\subsection{Scope and non-claims}\label{subsec:scope}

The analysis is descriptive and mechanism-aware. It does not establish legal finality in every jurisdiction, identify private deliberation, infer beneficial ownership from addresses, or rank institutional quality from speed alone. It does not assume that public market states exhaust either venue's internal review. It does not treat a Polymarket redemption event as total entitlement completion, because published upstream evidence leaves the balance-complete entitlement denominator blocked. It does not treat a Kalshi public finalization state as member-level withdrawable cash without separate account evidence.

The manuscript contains published Polymarket numerical inputs and qualified Kalshi-native finality inputs but no cross-venue numerical result. Matched outputs enter only after semantic and calendar certificates pass. This ordering is deliberate: a failure of semantic overlap, calendar support, or clock comparability narrows the result rather than changing the question after inspection.

\section{Related Work and Position in the Programme}\label{sec:related}

Prediction-market research traditionally studies information aggregation, calibration, contract design, and price interpretation \parencite{hanson2003,manski2006,wolfers2004}. Those literatures explain why event claims can aggregate beliefs, but they do not provide a common empirical definition of finality across heterogeneous institutional architectures.

Research on financial-market infrastructure distinguishes execution, clearing, settlement, finality, and loss allocation \parencite{cpmi2012}. Event contracts add a semantic layer: a rule-governed real-world proposition must be converted into an admissible payoff before financial settlement can occur. The relevant terminal path can therefore contain both information-governance states and accounting states.

Oracle research emphasizes the separation between external evidence, proposal, dispute, and final answer \parencite{eskandari2021}. The published Polymarket Parts I and II add exact adapter, protocol, payout, and holder-action layers. The Kalshi papers develop the corresponding venue-administered lifecycle and contractual-decidability layers. Paper 8 owns neither venue reconstruction. Its novelty is the cross-architecture identification problem: which functional coordinates can be compared, on which units, under which observation processes.

Statistically, the study draws on event-history and incomplete-data methods. Kaplan--Meier estimators handle right-censored single endpoints; Aalen--Johansen methods describe observed multistate transitions; Turnbull-type methods accommodate interval-censored times \parencite{kaplan1958,aalen1978,andersen1993,turnbull1976}. These tools are useful only after the institutional endpoint has been defined. A sophisticated estimator cannot make a proposal equivalent to a determination or recover a contestability path erased by a terminal snapshot.

The matching design is related to the broader observational-study literature \parencite{rosenbaum1983,stuart2010}, but the objective here is not a causal treatment effect of venue choice. Matching is used first to establish contractual comparability. Covariate balance does not repair payoff non-equivalence. Likewise, cluster-robust inference can account for repeated listings or shared events, but it does not turn child markets into independent real-world events \parencite{cameron2015}.

Within the Event-Linked Perpetuals programme, Paper 5 Part I owns Polymarket oracle adjudication, rule and clarification history, request generations, and adapter terminality. Part II owns exact question-to-condition mapping, protocol payout recording, payout vectors, observed redemption, and the entitlement boundary. Paper 7.1 owns Kalshi lifecycle and finalization; Paper 7.3 owns source and decidability clocks. Paper 6 uses these states for margin and conversion; Paper 9 uses leg-specific clocks for asynchronous spreads \parencite{p5a2026,p5b2026,p71_2026,p73_2026,p6_2026,p9_2026}.

The formal results below deliberately identify their ownership. Path non-identification under terminal coarsening is inherited from the venue-native studies. The new contribution here is the comparison layer: functional endpoint mapping, match certificates, calendar support, observation-compatible paired estimands, and multidimensional ordering.

\section{Venue-Native Mechanisms and Published Evidence}\label{sec:native}

A comparison begins with two native histories, not with a shared vocabulary imposed in advance.

\subsection{Polymarket: oracle, adapter, protocol, and holder layers}\label{subsec:poly-native}

Let a Polymarket contract history contain the following functional states:
\begin{enumerate}
  \item a versioned rule and ancillary-data state;
  \item technical request creation and one or more request generations;
  \item first proposal;
  \item dispute, reset, successor generation, or re-proposal where observed;
  \item final Oracle value;
  \item adapter consumption and terminal emission;
  \item Conditional Tokens payout recording, which makes the ordered payout rule effective;
  \item technical redeemability;
  \item observed redemption of any amount and observed positive-payout redemption.
\end{enumerate}

These states can share a transaction or block timestamp without being the same state. Canonical chain identity and log order can preserve ordering even when wall-clock separation is zero. Conversely, the published Part II evidence does not claim a universal same-transaction adapter--protocol ordering where publication-grade cross-contract keys are unavailable \parencite{p5b2026}.

Protocol payout recording is the Polymarket protocol-finality coordinate used in this paper. Observed redemption is a later holder action. It does not identify the fraction of all entitlement redeemed because total winning-position entitlement is not reconstructed at publication grade.

\newpage
\subsection{Kalshi: venue-administered determination and finalization}\label{subsec:kalshi-native}

The Kalshi lifecycle authority distinguishes market creation and activation, deactivation, close and close-time revision, first determination, contestability or amendment where publicly observable, and final settlement/finalization \parencite{p71_2026}. Public rules are distributed across market, event, series, source, and clarification surfaces. A current object is not automatically the historical object that governed an earlier determination.

Paper 7.3 adds a separate semantic clock construction. A named source can publish before the contract is uniquely decidable; later corrections can distinguish first from stable decidability \parencite{p73_2026}. These clocks are not inferred from the venue's determination timestamp, because doing so would define the primary source-to-venue delay as zero by construction.

Public finalization is a venue-level state. Member-level settlement or cash availability is a separate coordinate and can remain unmeasured in the population study. Ordinary and multivariate-event markets are separate families; they are not silently pooled when their creation process, message surface, or dependence structure differs.

\begin{figure}[t]
\centering
\begin{tikzpicture}[node distance=4mm and 4mm]
\node[pbox,text width=18mm] (p1) {Proposal};
\node[pbox,text width=18mm,right=of p1] (p2) {Dispute / reset};
\node[pbox,text width=18mm,right=of p2] (p3) {Oracle final};
\node[pbox,text width=18mm,right=of p3] (p4) {Protocol payout};
\node[pbox,text width=18mm,right=of p4] (p5) {Observed redemption};
\foreach \a/\b in {p1/p2,p2/p3,p3/p4,p4/p5}{\draw[parr] (\a)--(\b);}
\node[anchor=east,font=\footnotesize\bfseries] at ([xshift=-3mm]p1.west) {Polymarket};

\node[pbox,text width=18mm,below=13mm of p1] (k1) {Close};
\node[pbox,font=\scriptsize,text width=18mm,right=of k1] (k2) {First\\determination};
\node[pbox,text width=18mm,right=of k2] (k3) {Contest / amendment};
\node[pbox,text width=18mm,right=of k3] (k4) {Finalized};
\node[pbox,text width=18mm,right=of k4] (k5) {Member cash if observed};
\foreach \a/\b in {k1/k2,k2/k3,k3/k4,k4/k5}{\draw[parr] (\a)--(\b);}
\node[anchor=east,font=\footnotesize\bfseries] at ([xshift=-3mm]k1.west) {Kalshi};

\draw[pdash] (p1.south) -- (k2.north);
\draw[pdash] (p4.south) -- (k4.north);
\end{tikzpicture}
\caption{Venue-native paths. Dashed connectors align functions without claiming that the linked states have the same authority, challenge process, or observation surface.}
\label{fig:native-paths}
\end{figure}
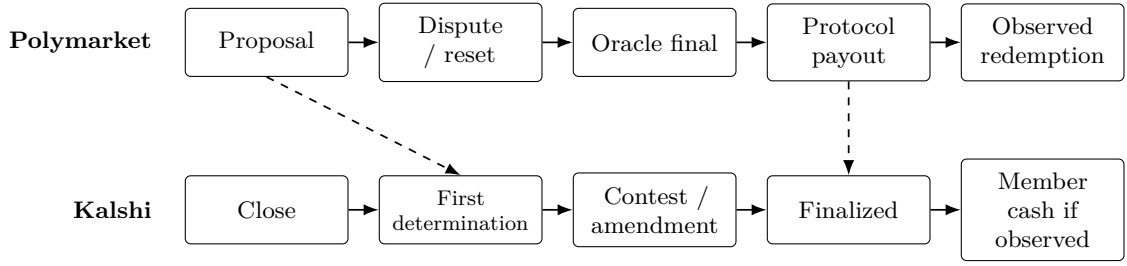

\subsection{Functional crosswalk}\label{subsec:crosswalk}

Table~\ref{tab:crosswalk} aligns roles, not names. Every row retains a mechanism-specific definition and can remain unmeasured on one side.

\begin{table}[H]
\centering
\caption{Functional crosswalk. Rows are comparison roles, not claims that the native states are institutionally identical.}
\label{tab:crosswalk}
\small
\begin{tabularx}{\textwidth}{@{}p{0.19\textwidth}XX@{}}
\toprule
Functional role & Polymarket native evidence & Kalshi native evidence \\
\midrule
Rule state & Ancillary data, rules, clarification versions, adapter configuration & Series/event sources, market rules, timing clauses, additional details \\
External evidence & Named public source where independently archived & Named settlement source and exact release object where archived \\
Contractual decidability & Generally unmeasured at population scale in the published Paper 5 evidence & First/stable decidability in Paper 7.3 adjudicated subset \\
First recorded result & First Oracle proposal for an eligible generation & First venue determination \\
Contestability & Dispute, reset, successor generation, re-proposal & Settlement timer, disputed/amended state, re-determination where observable \\
Final mechanism result & Successful final Oracle value & Last applicable determination after observed amendment/review \\
Venue/protocol terminal & Adapter terminality and CTF payout recording, kept separate & Finalized/settled venue market state \\
Cash realization & Observed CTF redemption; positive payout separately & Member settlement evidence if lawfully and population-compatibly observed \\
\bottomrule
\end{tabularx}
\end{table}

The first-result row is a useful example. A Polymarket proposal is submitted within an optimistic-oracle process; a Kalshi determination is venue-administered. They can occupy a comparable role as the first recorded result while differing in authority, access, and challenge process. A valid comparison states both the common role and the institutional difference.

\subsection{Evidence inputs}\label{subsec:evidence-inputs-native}
The comparison uses empirical native evidence on both sides, but the relevant uses remain asymmetric. Polymarket Parts I and II are published SSRN records. KMVE V1 defines Kalshi population and hierarchy, while the Kalshi lifecycle evidence supplies qualified endpoint and exact-clock facts. Paper 7.3 specifies the source/decidability record; no cross-venue matched-event registry is reported here.

\begin{table}[H]
\centering\footnotesize
\caption{Evidence inputs and their scientific uses in Paper 8.}
\label{tab:authority-status}
\begin{tabularx}{\textwidth}{@{}L{0.22\textwidth}L{0.23\textwidth}Y@{}}
\toprule
Layer & Evidence input & Scientific use \\
\midrule
Polymarket oracle/adapter & Published SSRN Part I & Native state definitions, denominators, path and observation-grade inputs \\
Polymarket protocol/redemption & Published SSRN Part II & Protocol finality, payout, redemption, censoring, and value-boundary inputs \\
Kalshi population/hierarchy & KMVE V1 & Native population, parent hierarchy, activity snapshot, and observation-surface inputs \\
Kalshi lifecycle/finality & Qualified Kalshi lifecycle evidence & Native cohorts, exact endpoint union, and exact determination-to-endpoint distributions; path, revision, and timer claims remain separately graded \\
Kalshi source/decidability & Paper 7.3 decidability design & Source/decidability comparisons only on a separately defined blind subset \\
Cross-venue match registry & No matched-event registry reported & No paired numerical claim \\
\bottomrule
\end{tabularx}
\end{table}

\section{Functional Finality Vector and Comparison Certificate}\label{sec:functional}

Let $v\in\{P,K\}$ index Polymarket and Kalshi. For a venue-native contract $i$, define the functional finality vector
\begin{equation}
\Phi_{i,v}=\left(
 t^{\mathrm{src}}_{i,v},
 t^{\mathrm{dec}}_{i,v},
 t^{\mathrm{close}}_{i,v},
 t^{\mathrm{first}}_{i,v},
 \mathcal I^{\mathrm{contest}}_{i,v},
 t^{\mathrm{mech}}_{i,v},
 t^{\mathrm{venue}}_{i,v},
 t^{\mathrm{cash}}_{i,v}
\right).
\label{eq:finality-vector}
\end{equation}
A coordinate can be absent when the mechanism has no distinct public state or when the evidence does not identify it.

\begin{definition}[Coordinate observation grade]\label{def:grade}
For coordinate $k$ of $\Phi_{i,v}$, let
\[
G_{i,v,k}\in\{E,I,S,Prx,U,X\},
\]
where $E$ is exact, $I$ interval-observed, $S$ snapshot-only, $Prx$ an explicit proxy, $U$ unmeasured, and $X$ conflicting. Right censoring is recorded separately from grade because an exact risk-set origin can coexist with an unobserved endpoint. For contestability, grades attach to each on/off endpoint and event mark; a set-level summary never upgrades its weakest load-bearing component.
\end{definition}

\begin{definition}[Functional endpoint map]\label{def:endpoint-map}
A functional endpoint map $M_v$ assigns a venue-native state or event to one coordinate of Equation~\eqref{eq:finality-vector} together with: (i) its economic role, (ii) its native authority, (iii) its observation grade, and (iv) the conditions under which the assignment is valid. The map does not rename the native state.
\end{definition}

\begin{definition}[Contestability episodes and occupation]\label{def:contestability-episodes}
For market $i$ on venue $v$, let
\[
\mathcal I^{\mathrm{contest}}_{i,v}
 =\left\{[t^{\mathrm{on}}_{i,v,r},t^{\mathrm{off}}_{i,v,r})\right\}_{r=1}^{R_{i,v}}
\]
be the set of publicly identifiable challenge or review episodes under the venue-native rules. An episode begins when an ordinary challenge/review opportunity becomes active or is first observed. Its closing endpoint is identified when that opportunity closes, is resolved, or is reset; if the observation window ends first, the closing endpoint remains right-censored. Its occupation is
\[
\Omega^{\mathrm{contest}}_{i,v}
 =\sum_{r=1}^{R_{i,v}}
   \left(t^{\mathrm{off}}_{i,v,r}-t^{\mathrm{on}}_{i,v,r}\right),
\]
with right-censored terms retained as intervals. A dispute, amendment, reset, or re-determination is a mark within the episode history; it is not by itself the complete contestability interval. If an event-time analysis requires one scalar anchor, it uses the explicitly labelled first activation $t^{\mathrm{contest,on}}_{i,v,1}$, never an unspecified generic contestability timestamp.
\end{definition}

A rule-defined timer can establish $t^{\mathrm{on}}$ at the first result only when the applicable rule version and timer identity are archived. When only a dispute or amendment mark is observed, the episode receives the corresponding exact or interval grade rather than a fabricated opening or closing time. Multiple dispute/reset or determination/amendment cycles remain separate episodes. The episode set is coded empty only when the applicable rule history establishes that no ordinary challenge/review opportunity exists; absence of an observed dispute mark is insufficient.

\begin{definition}[Comparison certificate]\label{def:certificate}
For candidate pair $(i_P,i_K)$ and duration coordinates $(a,b)$, a comparison certificate is
\[
\mathcal C_{i_P,i_K}^{a,b}=
\left(T^{\mathrm{sem}},M_P,M_K,G^{a,b},W^{\mathrm{cal}},R^{\mathrm{cens}},U^{\mathrm{inf}}\right),
\]
where $T^{\mathrm{sem}}$ is the semantic-match tier, $G^{a,b}$ the four-tuple of endpoint grades for $(a,b)$ on both venues, $W^{\mathrm{cal}}$ calendar-support class, $R^{\mathrm{cens}}$ the common censoring rule, and $U^{\mathrm{inf}}$ the inference unit or cluster. A paired duration is publication-admissible only when its certificate meets the stated conditions for that estimand.
\end{definition}

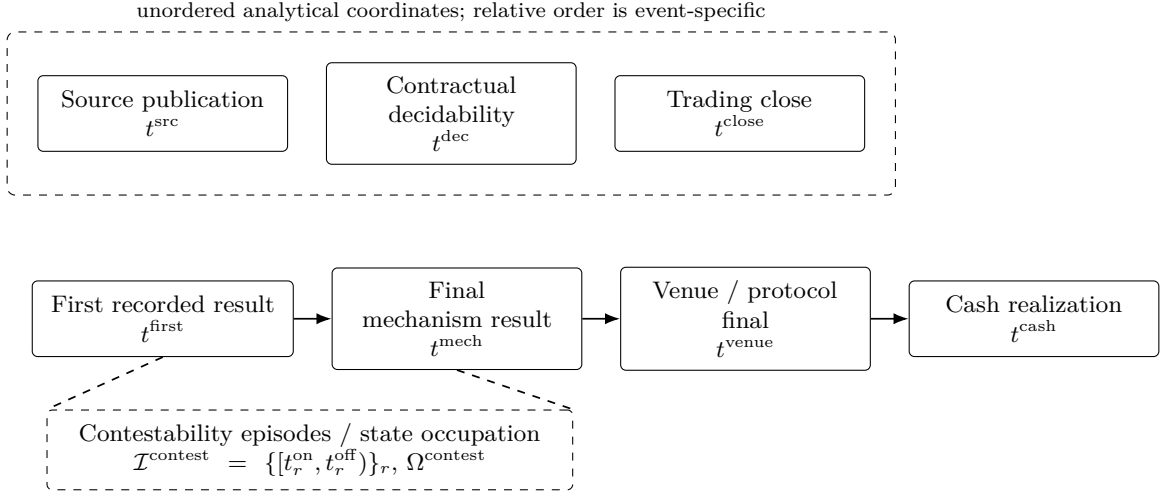
\begin{figure}[t]
\centering
\begin{tikzpicture}[node distance=5mm and 5mm]
\node[pwide,text width=29.5mm] (s) {Source publication\\$t^{\mathrm{src}}$};
\node[pwide,text width=29.5mm,right=of s] (d) {Contractual decidability\\$t^{\mathrm{dec}}$};
\node[pwide,text width=29.5mm,right=of d] (c) {Trading close\\$t^{\mathrm{close}}$};
\node[draw,dashed,rounded corners=3pt,fit=(s)(d)(c),inner sep=4mm,
      label={[font=\scriptsize,align=center]above:unordered analytical coordinates; relative order is event-specific}] (pre) {};

\node[pwide,below=17mm of s] (f) {First recorded result\\$t^{\mathrm{first}}$};
\node[pwide,text width=29.5mm,right=of f] (m) {Final\\mechanism result\\$t^{\mathrm{mech}}$};
\node[pwide,text width=29.5mm,right=of m] (v) {Venue / protocol\\final\\$t^{\mathrm{venue}}$};
\node[pwide,text width=29.5mm,right=of v] (h) {Cash realization\\$t^{\mathrm{cash}}$};
\foreach \a/\b in {f/m,m/v,v/h}{\draw[parr] (\a)--(\b);}

\node[psoft,text width=66mm,below=12mm of $(f)!0.5!(m)$] (x)
 {Contestability episodes / state occupation\\$\mathcal I^{\mathrm{contest}}=\{[t^{\mathrm{on}}_r,t^{\mathrm{off}}_r)\}_r$, $\Omega^{\mathrm{contest}}$};
\draw[dashed,line width=0.65pt] (f.south) -- (x.north west);
\draw[dashed,line width=0.65pt] (m.south) -- (x.north east);

\end{tikzpicture}
\caption{Common functional finality representation. Source publication, contractual decidability, and trading close are shown as an unordered analytical row: their relative order is event- and rule-specific. Downstream arrows encode functional dependence when the corresponding states are distinct, while contestability is represented by one or more graded episodes and their occupation time rather than by a single universal timestamp.}
\label{fig:finality-vector}
\end{figure}

\newpage
\subsection{Calendar-support classes}\label{subsec:calendar-support}

Calendar support is necessary because the two upstream populations are not automatically contemporaneous. Define:
\begin{itemize}
  \item $W=A$: both venue histories for the same economic event are observed under compatible path-level evidence within a common calendar window;
  \item $W=B$: the economic event is matched, but one venue has only terminal or interval-qualified evidence for one or more required coordinates;
  \item $W=C$: the venue populations are from non-overlapping windows or different event sets and support only standardized unpaired analysis;
  \item $W=X$: calendar or population identity is conflicting.
\end{itemize}

The prospective Kalshi lifecycle cohort begins after the published Polymarket follow-up boundary. It is therefore $W=C$ relative to the published Polymarket population unless a separate concurrent Polymarket cohort exists. A bounded historical Kalshi tranche can generate overlap candidates, but terminal historical records do not become full lifecycle histories merely because their dates overlap.

\subsection{Coordinate deletion and reduced vectors}\label{subsec:reduced-vector}

Not every comparison needs all eight coordinates. Let $J\subseteq\{1,\ldots,8\}$ and write $\Phi^J$ for the projection onto coordinates in $J$. A reduced comparison is admissible when every omitted coordinate is outside the estimand and the interpretation does not rely on it. For example, a comparison of first-result-to-venue-final duration does not require member cash time. By contrast, a claim about total capital lock cannot delete cash realization if the obligation remains open until cash.

This rule prevents a common form of hidden imputation: dropping a missing coordinate while retaining the economic interpretation that required it.

\subsection{Native clock graphs}\label{subsec:native-graphs}

The representation is not assumed to be a simple strict chain. Source publication can occur before or after close. A rule correction can revise the effective evidence set. Polymarket protocol finality and adapter terminal emission can share one transaction with an order identified only by log index. Kalshi determination and finalization can share one-second timestamp precision while receive order is not valid-time order. The analytical object is therefore a directed, partially ordered clock graph with set-valued contestability episodes rather than a sequence inferred from row order.

\section{Comparability and Non-Identification Results}\label{sec:formal}

\begin{theorem}[Endpoint-aligned paired comparability]\label{thm:endpoint-alignment}
Let $(i_P,i_K)$ be a semantically admissible pair and define venue-specific durations
\[
D_{i,v}^{a,b}=t_{i,v}^{b}-t_{i,v}^{a}.
\]
The paired contrast
\[
\Delta_i^{a,b}=D_{i,K}^{a,b}-D_{i,P}^{a,b}
\]
has a common functional interpretation only if: (i) $M_P(a)$ and $M_K(a)$ perform the same registered start function; (ii) $M_P(b)$ and $M_K(b)$ perform the same registered end function; (iii) endpoint grades meet the estimand-specific minimum; (iv) calendar support and censoring are compatible; and (v) the pair has the same payoff functional over the registered state space.
\end{theorem}

\begin{proof}
If either endpoint performs a different function, then the two durations integrate time over different institutional intervals. If grades differ incompatibly, one duration can be point-observed while the other represents a bound or terminal snapshot. If censoring or calendar support differs, selection into the observed pair can depend on venue-specific follow-up rather than mechanism duration. If payoff functions differ, the contracts need not become decidable or final under the same real-world states. Any one failure permits two data-generating histories with identical numerical durations but different economic meanings, so the scalar contrast is not identified as a common quantity.
\end{proof}

\begin{proposition}[Terminal coarsening]\label{prop:terminal-coarsening}
A mapping from a complete finality history to its terminal state, terminal value, and terminal timestamp is generally non-injective. Consequently, terminal-only matched records cannot identify path complexity, contestability frequency, or time spent in intermediate states.
\end{proposition}

\begin{proof}
Construct two histories with the same terminal value and timestamp. One proceeds directly from first result to finality. The other contains a dispute, reset or amendment, followed by restoration of the same final value before the same terminal timestamp. Their terminal projections coincide while their paths and state-occupation times differ.
\end{proof}

\begin{proposition}[Interval sign non-identification]\label{prop:interval-sign}
Suppose $D_{i,K}\in[\ell_K,u_K]$ and $D_{i,P}\in[\ell_P,u_P]$. Then
\[
\Delta_i\in[\ell_K-u_P,\;u_K-\ell_P].
\]
If this interval contains zero, the direction of the venue contrast is not identified. Midpoint subtraction can assign a sign not supported by the evidence.
\end{proposition}

\begin{proof}
The smallest feasible contrast combines the smallest Kalshi duration with the largest Polymarket duration; the largest feasible contrast combines the reverse endpoints. Any midpoint is an additional location model inside both intervals. When zero lies in the identified set, both venue orderings remain compatible with the observations.
\end{proof}

\begin{proposition}[Calendar-overlap insufficiency]\label{prop:calendar}
Exact semantic equivalence of two contract specifications does not by itself identify a paired mechanism contrast when their lifecycle evidence is observed in non-overlapping venue populations. In that case the comparison is unpaired unless an exact common event is observed on both venues with compatible histories.
\end{proposition}

\begin{proof}
Semantic equivalence identifies the payoff mapping, not the realized pair membership. If the venue populations concern different events or different calendar periods, the observed durations can differ because of event composition, operational regime, or observation design. Without a common event identity or a standardization model, the within-pair contrast is undefined.
\end{proof}

\begin{proposition}[Complete-case ranking reversal]\label{prop:complete-case}
Let $R_v$ indicate that both required clocks are observed on venue $v$. A comparison of $\E[D_v\mid R_v=1]$ can reverse the ranking of $\E[D_v]$ when clock observability depends differently on duration, path, or event type across venues.
\end{proposition}

\begin{proof}
Choose two strata, short and long, and let venue $P$ observe mostly long cases while venue $K$ observes mostly short cases even though within each stratum $P$ is faster. The complete-case mixture can then rank $K$ faster. Differential observation changes mixture weights, so the ranking is not invariant without a missingness assumption or common-support adjustment.
\end{proof}

\begin{proposition}[Governance-unit non-equivalence]\label{prop:governance-unit}
Counts of proposers, disputers, execution addresses, venue determinations, or review actions are not directly comparable governance concentration measures unless the unit of authority and denominator are aligned.
\end{proposition}

\begin{proof}
An address can be a relayer or custodian for several actors; a venue action can aggregate several internal reviewers; and a market can generate several externally visible actions. Equal counts therefore need not represent equal decision-maker concentration. A comparable concentration measure requires a common governance unit or separate native metrics with explicit interpretation.
\end{proof}

\begin{proposition}[Unpaired composition decomposition]\label{prop:decomposition}
Let $X$ denote pre-specified comparable strata and $\mu_v(x)=\E[D_v\mid X=x]$. For any reference distribution $w(x)$,
\[
\E[D_K]-\E[D_P]
=
\sum_x w(x)\{\mu_K(x)-\mu_P(x)\}
+
\sum_x\{p_K(x)-w(x)\}\mu_K(x)
-
\sum_x\{p_P(x)-w(x)\}\mu_P(x).
\]
The first term is a standardized mechanism-associated contrast; the remaining terms are composition components.
\end{proposition}

\begin{proof}
Add and subtract $\sum_x w(x)\mu_K(x)$ and $\sum_x w(x)\mu_P(x)$ from the two venue means and regroup.
\end{proof}

\begin{theorem}[No venue-independent scalar ordering]\label{thm:no-total-order}
Let a finality outcome be the vector
\[
Y_v=(\text{latency},\text{contestability episodes/occupation},\text{revision},\text{observability},\text{capital lock}).
\]
If neither venue weakly dominates the other in every coordinate, no venue-independent total ordering follows without an explicitly chosen weighting or social objective.
\end{theorem}

\begin{proof}
When one venue is better on one coordinate and worse on another, positive weight vectors can be chosen to rank either venue first. The ordering therefore depends on the normative weights. Only coordinatewise dominance is invariant to all monotone positive-weight aggregations.
\end{proof}

\begin{proposition}[Native-state schedule non-portability]\label{prop:schedule}
A risk-control schedule triggered by a native label on one venue is not portable to another venue merely by substituting the nearest label. Portability requires the same functional state and sufficiently timely observation.
\end{proposition}

\begin{proof}
Two labels can differ in whether a result exists, whether contestability remains, whether the payout rule is effective, or whether the claim is redeemable. A control applied at different functional states changes the exposure interval. Delayed observation can change it again. Functional and observational alignment are therefore necessary.
\end{proof}

\newpage
\section{Payoff-Preserving Matched-Event Design}\label{sec:matching}

The primary paired unit is not a title pair. It is a pair of contracts whose payoff functionals agree over the registered admissible state space.

\begin{definition}[Contract payoff functional]\label{def:payoff}
For contract $i$ on venue $v$, let
\[
\pi_{i,v}:\Omega_i\rightarrow\mathcal Y_{i,v}
\]
map real-world states and admissible source histories into the contract's terminal payout. The domain includes states needed to evaluate threshold, timing, correction, postponement, cancellation, invalidity, and rounding clauses, not only the realized state.
\end{definition}

\begin{criterion}[Tier-A semantic admissibility]\label{crit:tier-a}
A candidate pair is Tier A only if the evidence establishes equivalence of:
\begin{enumerate}
  \item real-world proposition and polarity;
  \item unit, threshold, comparison operator, and rounding;
  \item measurement interval, cut-off, and time zone;
  \item source identity or a documented source-equivalence rule;
  \item treatment of preliminary, revised, corrected, and final releases;
  \item postponement, cancellation, invalidity, void, and exceptional-event rules;
  \item terminal payoff over all registered admissible states;
  \item applicable rule versions before the outcome became final.
\end{enumerate}
Ex post equality of the realized answer is necessary but not sufficient.
\end{criterion}

\newpage
\subsection{Tier definitions}\label{subsec:tiers}

\begin{table}[H]
\centering
\caption{Semantic match tiers and permitted use.}
\label{tab:tiers}
\small
\begin{tabularx}{\textwidth}{@{}p{0.10\textwidth}p{0.47\textwidth}X@{}}
\toprule
Tier & Definition & Permitted analysis \\
\midrule
A & Payoff-functional equivalence and archived rule support satisfy Criterion~\ref{crit:tier-a}. & Membership in $\mathcal M^A_{\mathrm{sem}}$; exact or interval paired use only after the corresponding downstream conditions are met. \\
B & One documented nontrivial difference remains, but the pair supports a pre-specified sensitivity or bounding analysis. & Sensitivity, partial identification, or signed scenario analysis; not primary pooled matching. \\
C & Same broad event family or realized answer, with material semantic differences or incomplete rule support. & Unpaired institutional distributions and match-attrition reporting only. \\
Rejected & Different proposition, polarity, payoff support, or unresolved identity. & No cross-venue finality comparison. \\
\bottomrule
\end{tabularx}
\end{table}

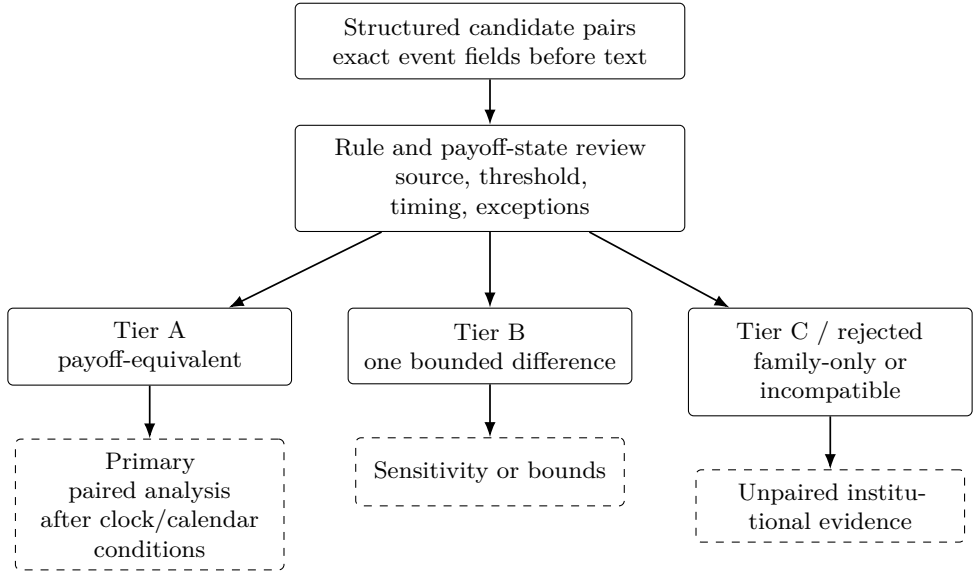
\begin{figure}[t]
\centering
\begin{tikzpicture}[node distance=6mm]
\node[pwide,text width=48mm] (cand) {Structured candidate pairs\\exact event fields before text};
\node[pwide,text width=48mm,below=of cand] (rules) {Rule and payoff-state review\\source, threshold, timing, exceptions};
\node[pwide,text width=34mm,below=10mm of rules,xshift=-45mm] (a) {Tier A\\payoff-equivalent};
\node[pwide,text width=34mm,below=10mm of rules] (b) {Tier B\\one bounded difference};
\node[pwide,text width=34mm,below=10mm of rules,xshift=45mm] (c) {Tier C / rejected\\family-only or incompatible};
\draw[parr] (cand)--(rules);
\draw[parr] (rules)--(a);
\draw[parr] (rules)--(b);
\draw[parr] (rules)--(c);
\node[psoft,text width=32mm,below=7mm of a] (pa) {Primary paired analysis\\after clock/calendar conditions};
\node[psoft,text width=32mm,below=7mm of b] (pb) {Sensitivity or bounds};
\node[psoft,text width=32mm,below=7mm of c] (pc) {Unpaired institutional evidence};
\draw[parr] (a)--(pa);\draw[parr] (b)--(pb);\draw[parr] (c)--(pc);
\end{tikzpicture}
\caption{Payoff-preserving match tiers. A common title or realized answer can generate a candidate but cannot promote a pair to Tier A.}
\label{fig:matching-tiers}
\end{figure}

\subsection{Candidate generation}\label{subsec:candidate-generation}

Candidate generation uses exact structured fields before text. Registered keys include event category, scheduled date, jurisdiction, measured variable, unit, threshold, source, and terminal payoff support. Text similarity can rank candidates for review but cannot create an accepted relation. Let $\mathcal M_{\mathrm{candidate}}$ denote the resulting structured candidate registry before semantic adjudication.

The candidate graph can be many-to-many. One economic event can have several strikes or duplicated listings on one venue. The pipeline therefore separates three identities:
\begin{enumerate}
  \item economic event identity;
  \item venue contract identity;
  \item comparison-pair identity.
\end{enumerate}
A market count is not an event count, and several child contracts do not create independent real-world outcomes.

\subsection{Double-pass adjudication}\label{subsec:double-pass}

Every Tier-A candidate receives two independent semantic reviews before reconciliation. Reviewers see frozen rules and source records but not the downstream duration contrast. The adjudication record stores initial decisions, disagreement fields, reconciliation, and exact evidence versions. Agreement is reported by field and overall tier; a reconciled match never erases initial disagreement.

\subsection{Temporal-overlap registry}\label{subsec:overlap-registry}

Each accepted semantic pair is then joined to an overlap registry. The registry records whether both venue histories cover the same event, whether the required clocks are path-level or terminal-only, and whether acquisition regimes are compatible. This produces four practical groups:
\begin{itemize}
  \item Tier A, path-complete, common-window pairs: $\mathcal M_{\mathrm{pair,exact}}$ for the interval whose endpoints are exact;
  \item Tier A, partially observed common-window pairs: $\mathcal M_{\mathrm{pair,interval}}$ for bounds and sign identification;
  \item Tier A specifications without common event histories: $\mathcal M_{\mathrm{sem}}$ evidence for a future concurrent study, not current paired data;
  \item Tier B/C contracts: sensitivity or unpaired cohorts, with their tier never hidden.
\end{itemize}

The published Polymarket population and prospective Kalshi cohort do not become paired through topic similarity. A matched-event result requires an actual common event in a shared evidence window.

\section{Identification and Statistical Protocol}\label{sec:identification}

The paper distinguishes paired, interval-qualified, and standardized unpaired estimands. They answer different questions and are never merged into one headline number.

\newpage
\subsection{Paired duration estimands}\label{subsec:paired-estimands}

For pair $i\in\mathcal M_{\mathrm{pair,exact}}^{a,b}$ with an accepted certificate, define
\begin{align}
D^{\mathrm{first\rightarrow mech}}_{i,v}
&=t^{\mathrm{mech}}_{i,v}-t^{\mathrm{first}}_{i,v},\\
D^{\mathrm{mech\rightarrow venue}}_{i,v}
&=t^{\mathrm{venue}}_{i,v}-t^{\mathrm{mech}}_{i,v},\\
D^{\mathrm{first\rightarrow venue}}_{i,v}
&=t^{\mathrm{venue}}_{i,v}-t^{\mathrm{first}}_{i,v},\\
D^{\mathrm{dec\rightarrow first}}_{i,v}
&=t^{\mathrm{first}}_{i,v}-t^{\mathrm{dec}}_{i,v},
\end{align}
where the last estimand exists only for the source/decidability-qualified subset. The paired contrast is $\Delta_i^D=D_{i,K}-D_{i,P}$.

The primary summary is the distribution of within-pair contrasts, not the difference between two independently estimated venue medians. Reported coordinates include the median, quartiles, tail shares over registered thresholds, and signed-rank or randomization-style diagnostics where the matched units justify them. Inference clusters at the economic-event or source-release level, not the child-market level.

\subsection{Interval-qualified pairs}\label{subsec:interval-pairs}

For $i\in\mathcal M_{\mathrm{pair,interval}}^{a,b}$, one or both durations are intervals and the pair contributes an identified contrast interval by Proposition~\ref{prop:interval-sign}. The paper reports:
\begin{itemize}
  \item the number of pairs with sign-identified positive, negative, and ambiguous contrasts;
  \item lower and upper bounds on threshold shares;
  \item interval widths by venue, endpoint, and event class;
  \item sensitivity to exact-only versus exact-plus-interval cohorts.
\end{itemize}
No midpoint is inserted in the primary analysis.

\subsection{Right censoring and multistate paths}\label{subsec:censoring}

Venue-native survival distributions use the correct risk-set origin for each functional transition. A market that has not reached its endpoint by the common administrative boundary is right-censored. Kaplan--Meier curves are used for one terminal endpoint when their assumptions are defensible; Aalen--Johansen estimates describe observed transitions among first result, contested/revised, mechanism-final, and venue-final states \parencite{kaplan1958,aalen1978,andersen1993}.

Cross-venue survival contrasts require the same functional origin and endpoint. A Polymarket protocol-final-to-redemption curve is not compared directly with a Kalshi determined-to-finalized curve: one begins after a payout rule is recorded and ends at holder action, while the other spans venue review and settlement.

\subsection{Standardized unpaired estimands}\label{subsec:standardized}

When exact event pairing is unavailable, the paper reports native distributions and a standardized mechanism-associated contrast over a declared common-support population. Pre-specified standardization variables can include event class, binary/scalar form, source class, scheduled horizon, early-close eligibility, rule complexity, and calendar regime. The reference weights $w(x)$ are specified before outcome inspection.

The resulting contrast is descriptive. It does not identify a causal effect of assigning an event to a venue, because listing choice, product design, source policy, and event selection are endogenous. Proposition~\ref{prop:decomposition} separates standardized within-stratum difference from composition.

\subsection{Common support and attrition}\label{subsec:support}

Every result begins with a denominator ladder:
\[
N_{\mathrm{candidate}}
\rightarrow N_{\mathrm{sem}}
\rightarrow \left(N_{\mathrm{pair,exact}},N_{\mathrm{pair,interval}}\right)
\rightarrow N_{\mathrm{analysis}}.
\]
The semantic, exact-pair, and interval-pair denominators are displayed separately; attrition reasons are mutually exclusive within each transition. The paper reports how accepted pairs differ from rejected or unavailable candidates. A small exact cohort is not hidden behind a large candidate count.

\subsection{Dependence and multiplicity}\label{subsec:dependence}

Several child markets can share one event, one source release, one determination second, or one oracle request generation. The default inference cluster is the narrowest upstream unit that can generate common finality shocks. For matched macro releases this can be the source release; for event families it can be the economic event. Collection or series clustering is a sensitivity when MVE instantiation creates many related contracts.

\subsection{Negative controls}\label{subsec:negative-controls}

The registered negative controls are:
\begin{enumerate}
  \item title-only or realized-answer-only pseudo-matches;
  \item endpoint-label substitution without functional mapping;
  \item shuffled venue assignment within broad event class;
  \item unrelated source releases matched on calendar time;
  \item complete-case estimates compared with full attrition and bounds;
  \item terminal-snapshot path reconstruction compared with prospectively observed histories;
  \item artificial midpoint imputation for interval clocks, reported only to quantify induced sign changes.
\end{enumerate}
A headline result is weakened if equally strong structure appears under a deliberately invalid match or endpoint definition.

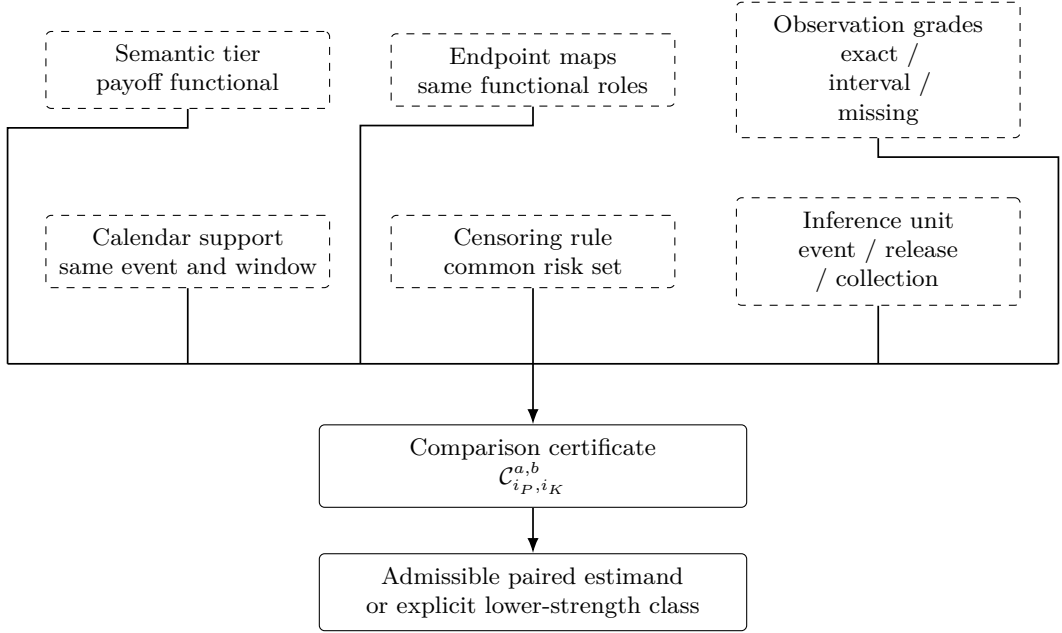
\begin{figure}[t]
\centering
\begin{tikzpicture}[node distance=6mm and 8mm]
\node[psoft,text width=34mm] (sem) {Semantic tier\\payoff functional};
\node[psoft,text width=34mm,right=of sem] (map) {Endpoint maps\\same functional roles};
\node[psoft,text width=34mm,right=of map] (grade) {Observation grades\\exact /\\ interval /\\ missing};
\node[psoft,text width=34mm,below=14mm of sem] (cal) {Calendar support\\same event and window};
\node[psoft,text width=34mm,right=of cal] (cens) {Censoring rule\\common risk set};
\node[psoft,text width=34mm,right=of cens] (unit) {Inference unit\\event / release / collection};

\coordinate (busY) at ($(cens.south)+(0,-10mm)$);
\coordinate (busL) at ($(cal.west |- busY)+(-5mm,0)$);
\coordinate (busR) at ($(unit.east |- busY)+(5mm,0)$);
\coordinate (gapOne) at ($(cal.east)!0.5!(cens.west)$);
\coordinate (gapTwo) at ($(cens.east)!0.5!(unit.west)$);
\draw[line width=0.65pt] (busL) -- (busR);
\foreach \n in {cal,cens,unit}{\draw[line width=0.65pt] (\n.south) -- (\n.south |- busY);}
\draw[line width=0.65pt] (sem.south) -- ++(0,-2.5mm) -| (busL);
\draw[line width=0.65pt] (map.south) -- ++(0,-2.5mm) -| (gapOne |- busY);
\draw[line width=0.65pt] (grade.south) -- ++(0,-2.5mm) -| (busR);

\node[pwide,text width=53mm,below=18mm of cens] (cert) {Comparison certificate\\$\mathcal C_{i_P,i_K}^{a,b}$};
\draw[parr] (cens.south |- busY) -- (cert.north);
\node[pwide,text width=53mm,below=of cert] (est) {Admissible paired estimand\\or explicit lower-strength class};
\draw[parr] (cert.south)--(est.north);
\end{tikzpicture}
\caption{Comparison certificate. Comparability is a conjunction of semantic, endpoint, evidence, calendar, censoring, and dependence conditions. Connectors terminate at node borders and do not imply a scalar score.}
\label{fig:certificate}
\end{figure}

\section{Multidimensional Finality Outcomes}\label{sec:outcomes}

Latency is one coordinate of finality, not its definition. The empirical comparison reports a vector of outcomes and preserves trade-offs.

\subsection{Latency and state occupation}\label{subsec:latency-occupation}

For functional state $s$, define occupation
\[
\Omega_{i,v,s}(a,b)=\int_a^b \ind\{S_{i,v}(t)=s\}\,dt.
\]
Two markets can share first and last timestamps while spending different time in contestability, amendment, protocol-final-but-unredeemed, or other intermediate states. State occupation therefore complements scalar durations.

\newpage
\subsection{Path complexity}\label{subsec:path-complexity}

Path complexity is reported through exact native marks rather than one synthetic index alone. Components include:
\begin{itemize}
  \item number of request generations or determinations;
  \item observed disputes, resets, amendments, or re-determinations;
  \item close-time revisions or reopenings;
  \item number of rule/source versions relevant to the path;
  \item number of distinct authority transitions;
  \item terminal-only or left-truncated coverage class.
\end{itemize}
A composite complexity score, if used, is secondary and sensitivity-tested across weights.

\subsection{Contestability and revision}\label{subsec:contestability}

Contestability is not measured by the same native event on both venues. Polymarket exposes rule-defined proposal liveness, disputes, resets, and successor generations. Kalshi exposes a determination/finalization path and, where observable, a settlement-timer window, disputed/amended states, or re-determination evidence. Following Definition~\ref{def:contestability-episodes}, the comparison reports:
\begin{enumerate}
  \item whether the applicable rules create an ordinary challenge or review opportunity;
  \item first identified activation $t^{\mathrm{contest,on}}_{1}$ and, where identified, every episode $[t^{\mathrm{on}}_r,t^{\mathrm{off}}_r)$;
  \item contestability occupation $\Omega^{\mathrm{contest}}$, number of episodes, and right-censored episode time;
  \item observed dispute, reset, amendment, or re-determination marks within a qualified denominator;
  \item whether the final result differs from the first recorded result;
  \item time from first result to final mechanism result.
\end{enumerate}
Absence of a public dispute mark is not coded as absence of contestability. A timer-defined opportunity, an observed dispute mark, and full interval occupation are separate evidence objects and receive separate grades.

\subsection{Rule and source observability}\label{subsec:rule-source}

For each venue, report the proportion of the analysis population with:
\begin{itemize}
  \item archived applicable rules before the relevant outcome;
  \item exact named source identity;
  \item exact release object;
  \item exact or interval publication time;
  \item exact or interval contractual-decidability time;
  \item version history through finality.
\end{itemize}
This is both a measurement-quality outcome and an institutional characteristic. A faster visible terminal state can coexist with less observable upstream evidence.

\subsection{Governance structure}\label{subsec:governance}

Governance comparison remains native-unit first. Polymarket can report execution-address concentration, proposer/disputer concentration, access configuration, and route composition subject to identity limits. Kalshi can report public venue actions, documented authority, and review-state incidence, but not necessarily individual internal decision makers. A common scalar governance concentration is not produced unless units and denominators are proven compatible.

\subsection{Venue/protocol finality and cash realization}\label{subsec:cash}

The terminal accounting layer differs sharply. On Polymarket, protocol payout recording makes claims redeemable and observed redemption records holder action. On Kalshi, public finalization represents venue settlement, while member-level settlement evidence may be authenticated and account-specific. The comparison therefore has two policies:
\begin{enumerate}
  \item compare venue/protocol terminal states when their economic role is aligned;
  \item compare cash realization only on a separate subset with compatible account-level evidence.
\end{enumerate}
The public population study does not impute Kalshi member cash from finalization or Polymarket holder cash from protocol finality.

\subsection{Capital lock}\label{subsec:capital-lock}

For a notional capital requirement $K_i$ and annualized service rate $r$, let $\Delta t_{i,v}^{a,b}=t_{i,v}^{b}-t_{i,v}^{a}$ and let $Y$ be one year in the same time unit as the timestamps. A simple sensitivity is
\[
C_{i,v}^{a,b}(r)=K_i\left[\exp\left\{r\frac{\Delta t_{i,v}^{a,b}}{Y}\right\}-1\right].
\]
The exponent is dimensionless. This is not a welfare estimate. It converts aligned duration coordinates into a common financing unit under explicit assumptions. State-specific variants apply different service rates to contestable, venue-final, and holder-controlled intervals.

\subsection{Pareto comparison}\label{subsec:pareto}

\Needspace{12\baselineskip}
Contestability and revision are not globally monotone goods, so the paper specifies two orientation profiles rather than declaring one universal ordering.
\begin{itemize}
  \item \emph{Speed--capital profile}: shorter aligned latency, shorter contestability occupation, shorter post-first-result revision time, lower capital lock, and higher observability are preferred.
  \item \emph{Reviewability profile}: an observable challenge opportunity and higher rule/source observability are preferred, while aligned latency and capital lock remain lower-is-better; revision incidence is reported but not assigned a monotone direction.
\end{itemize}
For each profile, dominance is certified only when every load-bearing coordinate has a compatible identified interval. Profile-specific dominance and non-dominance are reported separately. When no venue dominates under the registered profiles, the identified object is a trade-off surface rather than a winner.

\begin{figure}[t]
\centering
\begin{tikzpicture}[node distance=5mm and 10mm]
\node[pwide,text width=35mm] (lat) {Latency and state occupation};
\node[pwide,text width=35mm,right=of lat] (con) {Contestability and revision};
\node[pwide,text width=35mm,right=of con] (obs) {Rule/source observability};
\node[pwide,text width=35mm,below=of lat] (gov) {Governance-native structure};
\node[pwide,text width=35mm,right=of gov] (cash) {Venue/protocol and cash boundary};
\node[pwide,text width=35mm,right=of cash] (cap) {Capital-lock sensitivity};

\coordinate (busY) at ($(cash.south)+(0,-8mm)$);
\coordinate (busL) at ($(gov.west |- busY)+(-5mm,0)$);
\coordinate (busR) at ($(cap.east |- busY)+(5mm,0)$);
\coordinate (gapOne) at ($(gov.east)!0.5!(cash.west)$);
\coordinate (gapTwo) at ($(cash.east)!0.5!(cap.west)$);
\draw[line width=0.65pt] (busL) -- (busR);
\foreach \n in {gov,cash,cap}{\draw[line width=0.65pt] (\n.south) -- (\n.south |- busY);}
\draw[line width=0.65pt] (lat.south) -- ++(0,-2.2mm) -| (busL);
\draw[line width=0.65pt] (con.south) -- ++(0,-2.2mm) -| (gapOne |- busY);
\draw[line width=0.65pt] (obs.south) -- ++(0,-2.2mm) -| (busR);

\node[psoft,text width=58mm,below=18mm of cash] (pareto) {Coordinatewise comparison\\Pareto dominance or explicit trade-off};
\draw[parr] (cash.south |- busY) -- (pareto.north);
\end{tikzpicture}
\caption{Multidimensional outcome space. A scalar venue ranking is not produced unless one venue dominates or a normative weighting function is declared.}
\label{fig:multidimensional}
\end{figure}
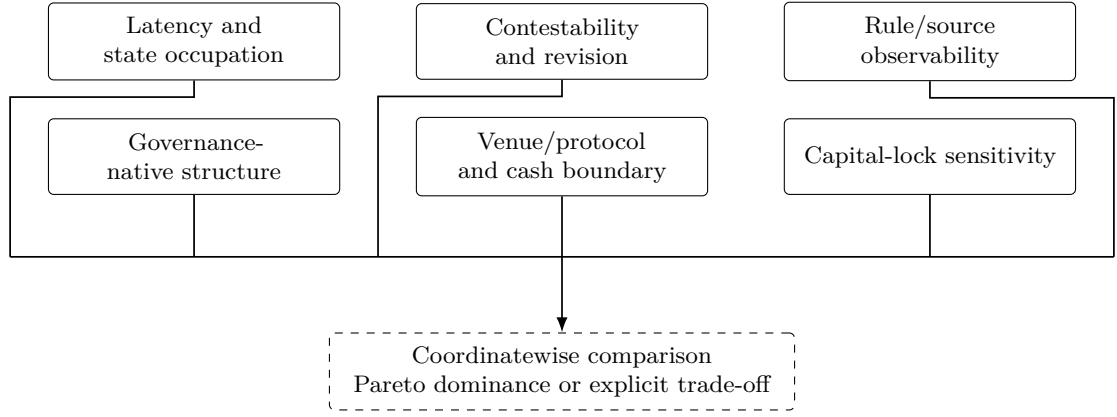

\section{Polymarket Evidence Layer}\label{sec:polymarket-input}

This section uses published Polymarket evidence as one-venue inputs. The numbers are not estimates of Kalshi and are not cross-venue differences.

\subsection{Oracle and adapter authority}\label{subsec:p5a-input}

Paper 5 Part I freezes 185,550 initialized adapter-question instances, 182,671 with at least one Oracle request, and 2,879 structural no-request questions. It reconstructs 184,148 request generations and 1,477 successor-generation rows. The exact stable-ID metadata relation matches 104,032 questions, leaves 81,518 unmatched, and has no ambiguous exact match in the documented relation \parencite{p5a2026}.

Published evidence measures request, proposal, dispute, reset, Oracle settlement, adapter terminality, and clarification timing. It does not measure population-wide source publication or contractual decidability. Paper 8 therefore treats those two Polymarket coordinates as unmeasured unless a separate qualified subset is produced.

\subsection{Protocol finality and observed redemption}\label{subsec:p5b-input}

Paper 5 Part II freezes an exact bridge of 108,638 linked conditions. Of these, 99,283 have an observed protocol-resolution event at the fixed snapshot. Within the resolved risk set, 92,158 have an observed redemption of any amount and 91,817 an observed positive-payout redemption. Condition-specific Kaplan--Meier medians from first protocol resolution are 182 seconds for any redemption and 200 seconds for positive-payout redemption \parencite{p5b2026}.

The exact-linked payout taxonomy contains 53,847 canonical $(0,1)$ vectors, 45,024 canonical $(1,0)$ vectors, 410 fifty-fifty vectors, two other valid vectors, and 9,355 conditions with no observed resolution. Cross-contract Oracle--adapter--protocol ordering is intentionally conservative: 823 conditions have an interval-qualified terminal generation, 48 have multiple candidate generations, 91,638 have no compatible terminal generation in the published evidence, and 16,129 are right-censored or otherwise unevaluable.

\begin{table}[H]
\centering
\caption{Published Polymarket evidence and its Paper 8 use. Different rows have different units and denominators.}
\label{tab:poly-inputs}
\small
\begin{tabularx}{\textwidth}{@{}XrX@{}}
\toprule
Published object & Count & Permitted comparative use \\
\midrule
Initialized adapter-question instances & 185,550 & Upstream population and request-path attrition \\
Questions with at least one Oracle request & 182,671 & Request-path denominator, not source-publication coverage \\
Exact-linked CTF conditions & 108,638 & Protocol linkage and payout/finality population \\
Protocol-resolved linked conditions & 99,283 & Protocol-finality and redemption risk set \\
Any redemption observed & 92,158 & Holder-action endpoint; not entitlement completion \\
Positive redemption observed & 91,817 & Positive holder-realization endpoint \\
Interval-qualified terminal generation & 823 & Interval-compatible cross-layer path subset \\
No compatible terminal generation & 91,638 & Explicit non-comparability/attrition, not zero delay \\
\bottomrule
\end{tabularx}
\end{table}

\subsection{What is not imported}\label{subsec:not-imported}

Paper 8 does not upgrade a Polymarket coordinate merely because Kalshi exposes a superficially similar field. Specifically:
\begin{itemize}
  \item request creation is not source publication or contractual decidability;
  \item Oracle settlement is not protocol payout recording;
  \item protocol payout recording is not holder cash realization;
  \item observed redemption is not entitlement completion;
  \item an identity join is not a causal timing relation;
  \item missing cross-contract order is not repaired by equal timestamps.
\end{itemize}

The published upstream G-VALUE boundary remains binding. Paper 8 can compare observed holder actions or capital-time sensitivities, but it cannot compare fractions of total economic entitlement completed without a new balance-complete denominator.

\section{Kalshi Native Finality and Temporal-Overlap Contract}\label{sec:kalshi-contract}

The Kalshi empirical programme supplies several distinct inputs. Paper 8 consumes each only for the coordinates it owns.

\subsection{Mandatory and optional upstreams}\label{subsec:mandatory-upstreams}

Paper 7.1 supplies the Kalshi lifecycle and finality evidence used here: defined cohorts, exact endpoint-source classes, exact determination-to-endpoint pairs, and observation-gap qualifications.

Paper 7.3 is mandatory only for comparisons beginning at source publication or contractual decidability. It must deliver blinded, reproducible source objects, first and stable decidability clocks, interval bounds, adjudication grades, and reviewer agreement. Markets without these clocks remain eligible for later functional intervals such as first-result-to-finalization.

Paper 7.2 supplies MVE population and dependence structure. It becomes load-bearing when MVE contracts enter the comparison because ticker multiplicity, collection identity, and shared underlying legs determine the inference unit. Paper 7.4 is not an upstream finality authority; it consumes Paper 8-compatible clocks for price analysis.

\subsection{Three Kalshi evidence populations}\label{subsec:three-populations}

The comparison must keep three Kalshi populations separate:
\begin{enumerate}
  \item a bounded historical terminal tranche, useful for overlap candidate generation and terminal fields but not full path frequencies;
  \item a prospective ordinary/MVE lifecycle cohort, useful for complete or gap-qualified paths but not automatically paired to the published Polymarket historical population;
  \item a source/decidability adjudication subset, useful for upstream semantic clocks but not representative unless its selection and attrition are modeled.
\end{enumerate}

This separation is essential. A terminal historical record can overlap the Polymarket calendar while lacking the path needed for contestability. A prospective path can be complete while referring to an event outside the published Polymarket follow-up. A decidability sample can have excellent clocks while being too small for venue-wide inference.

\subsection{Calendar-overlap requirement}\label{subsec:calendar-overlap-requirement}

Define $\mathcal E_P$ and $\mathcal E_K$ as economic-event identities represented in the accepted venue populations. Three nested objects are defined separately.

First, the semantic registry is
\[
\mathcal M_{\mathrm{sem}}
=\left\{(i,j): e(i)=e(j),\;T^{\mathrm{sem}}_{ij}\in\{A,B\}\right\},
\]
with the adjudicated tier retained on every row. Let $\mathcal M^{A}_{\mathrm{sem}}=\{(i,j)\in\mathcal M_{\mathrm{sem}}:T^{\mathrm{sem}}_{ij}=A\}$. The semantic registry records adjudication evidence; it is not yet a duration denominator.

For each registered functional interval $(a,b)$, the exact paired set is
\[
\begin{aligned}
\mathcal M^{a,b}_{\mathrm{pair,exact}}
=\bigl\{(i,j)\in\mathcal M_{\mathrm{sem}}:\;&
T^{\mathrm{sem}}_{ij}=A,\quad W^{\mathrm{cal}}_{ij}=A,\\
& G_{i,P,a}=G_{i,P,b}=G_{j,K,a}=G_{j,K,b}=E,\\
& \mathcal C_{ij}^{a,b}\text{ passes}\bigr\}.
\end{aligned}
\]
Exact here refers jointly to Tier-A semantics, common-event calendar support, and exact compatible endpoints.

The interval-qualified paired set is
\[
\begin{aligned}
\mathcal M^{a,b}_{\mathrm{pair,interval}}
=\bigl\{(i,j)\in\mathcal M_{\mathrm{sem}}:\;&
T^{\mathrm{sem}}_{ij}=A,\quad W^{\mathrm{cal}}_{ij}\in\{A,B\},\\
& G^{a,b}_{ij}\in\{E,I\}^{4},\\
& \mathcal C_{ij}^{a,b}\text{ satisfies the interval-comparison conditions}\bigr\}
\setminus\mathcal M^{a,b}_{\mathrm{pair,exact}}.
\end{aligned}
\]
The three denominators are never merged: $\mathcal M_{\mathrm{sem}}$ measures semantic admissibility, $\mathcal M_{\mathrm{pair,exact}}$ supports point-observed within-event contrasts, and $\mathcal M_{\mathrm{pair,interval}}$ supports bounds only.

The prospective Kalshi cohort belongs to either paired set only when the same economic events are represented in a compatible Polymarket evidence layer. Otherwise it contributes to the standardized unpaired comparison.

The bounded historical Kalshi tranche is the natural first source of overlap candidates with the published Polymarket population. Its terminal-only coverage is carried into each comparison certificate. A Tier-A semantic match with a missing path coordinate can remain in $\mathcal M_{\mathrm{sem}}$ and, when bounded endpoints exist, in $\mathcal M_{\mathrm{pair,interval}}$; it cannot enter $\mathcal M_{\mathrm{pair,exact}}$.

\newpage
\subsection{KMVE V1 population and observation evidence}\label{subsec:kalshi-r31-authority}

KMVE V1 defines the Kalshi MVE market-object population for the registered seven-day interval without replacing the final follow-up path. The population record contains 190 terminal-validated time-shard leaves and reproduces
\[
7{,}617{,}371=7{,}611{,}594+5{,}777,
\]
where 7,611,594 unique exact tickers fall inside the half-open REST creation-time interval and 5,777 overlap-envelope observations are retained but excluded. Exact duplicate groups and immutable identity conflicts are zero \parencite{kmve2026}.

\begin{table}[H]
\centering\footnotesize
\caption{Kalshi KMVE V1 native population and observation-surface inputs used by Paper 8.}
\begin{tabularx}{\textwidth}{@{}XrX@{}}
\toprule
Native object & Count & Permitted Paper 8 use \\
\midrule
REST enrollment MVE market objects & 7,611,594 & Native population scale and composition; not paired events \\
Observed exact event keys & 5,262,526 & Event dependence and composition \\
Current-validated collection identities & 3 & Current hierarchy; historical versions remain unmeasured \\
Current parent-derived series identities & 3 & Current collection-to-series relation; not a market-row field \\
Selected-leg primitives & 83,701 & MVE dependence and primitive-reuse structure \\
Exact structural signatures & 7,611,594 & Tuple uniqueness under the registered formula; not independent shocks \\
Current-snapshot activity-eligible objects & 2,692,787 & Snapshot descriptor; not historical trading \\
WS enrollment created notifications & 1,487,330 & Notification-surface denominator, not REST creation population \\
WS enrollment lifecycle tickers & 2,498,640 & Lifecycle-surface denominator, including left truncation \\
WS created $\cap$ REST & 276,177 & Exact ticker overlap without gap extrapolation \\
WS lifecycle $\cap$ REST & 378,129 & Exact ticker overlap without gap extrapolation \\
\bottomrule
\end{tabularx}
\label{tab:p8-r31-authority}
\end{table}

The cross-surface divergence is load-bearing. Only 18.57 percent of the WS created-notification set intersects the REST creation-time population, while the same intersection is 3.63 percent of REST membership. This is not a recall estimate: REST and WebSocket surfaces expose different evidence objects, and the WebSocket stream includes left-truncated histories and documented gaps. Paper 8 therefore treats population authority, lifecycle authority, and cross-surface classification as separate coordinates.

The closure correction independently validates the exact current collection keys through read-only GET responses and supplies their current series relations:
\begin{align*}
\texttt{KXMVECROSSCATEGORY-R} &\mapsto \texttt{KXMVECROSSCATEGORY},\\
\texttt{KXMVECROSSCATEGORY-SHARD1-R} &\mapsto \texttt{KXMVECROSSCATEGORY-SHARD1},\\
\texttt{KXMVESPORTSMULTIGAMEEXTENDED-R} &\mapsto \texttt{KXMVESPORTSMULTIGAMEEXTENDED}.
\end{align*}
These are post-freeze current-identity observations. They do not establish historical collection versions, historical series rules, or rule state at market creation.

The structural-signature audit excludes the output MVE ticker. It uses canonical JSON containing the exact event key, current market-row series field where present, collection ticker, and a lexicographically sorted multiset of exact selected-leg tuples $(\text{event ticker},\text{market ticker},\text{side})$. Duplicate legs are retained. The 7,611,594 market objects have 7,611,594 exact signatures and no collision group under that formula. Because the event identity is load-bearing, this result establishes exact object-structure uniqueness rather than 7.6 million independent economic propositions.

Activity is also coordinate-specific. All market rows expose parseable current-snapshot \texttt{volume\_fp}, \texttt{volume\_24h\_fp}, and \texttt{open\_interest\_fp}. Positive \texttt{volume\_fp} occurs for 2,692,787 objects; positive \texttt{open\_interest\_fp} for 2,671,402; \texttt{volume\_24h\_fp} is zero throughout; and the four registered trade-count aliases are absent throughout. The registered current-snapshot activity-eligible class is
\[
\mathbf 1\{\texttt{volume\_fp}>0\;\lor\;\texttt{open\_interest\_fp}>0\;\lor\;\text{first available trade-count alias}>0\},
\]
and contains 2,692,787 objects, or 35.38 percent. It is not a historical traded-market denominator and does not identify when activity occurred.

These hierarchy results constrain the inference unit when MVE markets enter Paper 8. Millions of child tickers cannot be treated as millions of independent event-finality observations. Collection, event, selected-leg, structural-signature, and source-release clustering remain separate. Series identity, activity history, rule/source history, and exact recovery of 1,211,153 WS-created-notification tickers absent from the REST union remain field-specific partial observations; no broad recovery census is substituted.

\subsection{Qualified Kalshi-native finality evidence}\label{subsec:kalshi-finality-authority}
The registered follow-up closes exact public endpoint status for the corrected native cohorts. The ordinary primary cohort contains 152,694 markets reconstructed as at risk at $T_0$. The MVE primary cohort is the 7,611,594-object REST enrollment population.

\begin{table}[H]
\centering\footnotesize
\caption{Kalshi-native finality inputs available to Paper 8. Counts are exact tickers and are not cross-venue estimates.}
\label{tab:kalshi-finality-inputs}
\begin{tabularx}{\textwidth}{@{}L{0.19\textwidth}L{0.31\textwidth}rY@{}}
\toprule
Family & Native object & Count & Permitted Paper 8 use \\
\midrule
Ordinary & Reconstructed $T_0$-at-risk cohort & 152,694 & Native ordinary risk set; exact enrollment births remain separate/unidentified \\\addlinespace[2pt]
Ordinary & Exact WS endpoint by follow-up & 71,657 & Venue-final endpoint count in the corrected stock \\\addlinespace[2pt]
Ordinary & No exact endpoint by follow-up & 81,037 & Native no-endpoint class, not member default \\\addlinespace[2pt]
Ordinary & Exact determination-to-endpoint pairs & 70,979 & Native exact-clock distribution; mean 322.50 s, max 7,133 s \\
\addlinespace
MVE & REST enrollment population & 7,611,594 & Native MVE market-object cohort \\\addlinespace[2pt]
MVE & REST exact endpoint only & 7,237,993 & Exact public finalization without an exact joined WS endpoint \\\addlinespace[2pt]
MVE & WS/REST exact agreement & 119,212 & Exact endpoint supported on both surfaces \\\addlinespace[2pt]
MVE & WS exact endpoint only & 371 & Exact lifecycle endpoint absent from the joined REST endpoint set \\\addlinespace[2pt]
MVE & No exact endpoint by follow-up & 254,018 & No admissible exact endpoint; current state is not back-projected \\\addlinespace[2pt]
MVE & Exact determination-to-endpoint pairs & 126,806 & Native exact-clock distribution; mean 9.676 s, max 15 s \\
\bottomrule
\end{tabularx}
\end{table}

The MVE endpoint union contains 7,357,576 exact public endpoints, or 96.66 percent of the cohort. Yet 7,230,770 exact finalizations lack an exact joined WebSocket path. In the ordinary cohort, 71,657 exact endpoints exceed the 70,979 exact determination-to-endpoint pairs by 678 markets. These are load-bearing observation-surface results: public endpoint finality can be substantially more complete than determination-path visibility.

The gap ledger reinforces the distinction. A gap intersects the registered path interval for 6,421,306 MVE objects and only eight ordinary cohort members. Two exact clocks remain point-identified despite an intervening gap, but absence of intermediate revisions does not. Current settlement-timer fields are present for every MVE object; historical version-consistent timer binding is not. Paper 8 can therefore compare aligned public endpoint roles and qualified exact durations, but it cannot treat the timer field as a historical contestability clock or infer member cash.

\newpage
\subsection{Additional evidence required for particular comparisons}\label{subsec:required-kalshi-artifacts}

The remaining cross-venue integration requires at minimum:
\begin{itemize}
  \item exact event, series, and collection parents where relevant;
  \item applicable rule/source versions and their retrieval/valid-time relation;
  \item determination and finalization event identities;
  \item dispute/amendment observability denominator;
  \item source/decidability packet registry for the qualified subset;
  \item gap, reconnect, reconciliation, and cross-surface classifications;
  \item documented table and figure derivations for any reported comparison.
\end{itemize}

\subsection{Observation-surface caution}\label{subsec:surface-caution}

WebSocket delivery order, REST retrieval order, and valid-time state order are different objects. Exact payload clocks can identify a transition while receipt order remains reversed or tied. A public lifecycle identifier can exist even when a later current REST lookup is unavailable. Paper 8 therefore imports the final Kalshi semantic ledger, not raw frame order or a current-availability proxy.

\section{Registered Results and Claim Contract}\label{sec:results-contract}

Kalshi-native endpoint and exact-clock results are available. Cross-venue numerical results require semantic, calendar, and endpoint alignment. Results appear in a fixed order so that attractive latency differences cannot precede their denominators and comparability evidence.

\subsection{Result order}\label{subsec:result-order}

\begin{enumerate}
  \item evidence and population table for each venue, including the defined Kalshi KMVE V1 population/observability crosswalk;
  \item candidate-pair generation and semantic attrition;
  \item Tier A/B/C adjudication and reviewer agreement;
  \item calendar-overlap and observation-grade matrix;
  \item endpoint-specific exact, interval, censored, and unavailable counts;
  \item exact paired duration/path contrasts and interval-qualified bounds, with separate denominators;
  \item standardized unpaired mechanism distributions and composition decomposition;
  \item contestability, revision, rule/source observability, and governance-native results;
  \item capital-lock sensitivities and Pareto comparison;
  \item negative controls, robustness, and claim downgrades.
\end{enumerate}

\newpage
\subsection{Version-locked empirical questions and decision rules}\label{subsec:hypotheses}

Questions for which no defensible ex ante scalar threshold exists are registered as estimands rather than directional hypotheses. The three hypotheses below have explicit \textsc{pass}/\textsc{fail}/\textsc{blocked} rules. No adjective such as ``substantial,'' ``nontrivial,'' or ``material'' acts as an unstated threshold.

\textbf{Q1: Match attrition.} Report
\[
r_A=\frac{|\mathcal M^{A}_{\mathrm{sem}}|}{|\mathcal M_{\mathrm{candidate}}|},\qquad
r_E^{a,b}=\frac{|\mathcal M^{a,b}_{\mathrm{pair,exact}}|}{|\mathcal M^{A}_{\mathrm{sem}}|},\qquad
r_I^{a,b}=\frac{|\mathcal M^{a,b}_{\mathrm{pair,interval}}|}{|\mathcal M^{A}_{\mathrm{sem}}|}.
\]
The disposition is \textsc{estimated} when all denominators and tier decisions rebuild, and \textsc{blocked} otherwise. No directional threshold is imposed.

\textbf{H2: Endpoint sensitivity.} For each registered interval $j$, let $\mathcal J_j$ be the identified uncertainty set for $\theta_j=D_{K,j}-D_{P,j}$ and define
\[
\mathcal O_j=
\begin{cases}
P\text{ faster}, & \inf \mathcal J_j>0,\\
K\text{ faster}, & \sup \mathcal J_j<0,\\
\text{no strict order}, & 0\in\mathcal J_j.
\end{cases}
\]
H2 \textsc{passes} if at least one registered interval is sign-identified as $P$ faster and another as $K$ faster. It \textsc{fails} if at least two intervals are sign-identified and all strict order classes agree. It is \textsc{blocked} otherwise.

\textbf{Q3: Observation-surface sensitivity.} Compare grade-aware reconstruction with terminal-only and complete-case reconstructions using the frozen vector
\[
A_{\mathrm{obs}}=
\left(\Delta N_{\mathrm{admissible}},\;
\operatorname{TV}(p_{\mathrm{path}}^{G},p_{\mathrm{path}}^{T}),\;
\Delta\widetilde D,\;
N_{\mathrm{order\ class\ changes}}\right).
\]
Here $\operatorname{TV}$ is total-variation distance between the two registered path-class distributions. The four components are reported without a post hoc binary threshold. The disposition is \textsc{estimated} or \textsc{blocked} according to the observation and reproducibility conditions.

\textbf{Q4: Composition.} For every standardized unpaired duration family, report the raw contrast $\theta_{\mathrm{raw}}$, standardized contrast $\theta_{\mathrm{std}}$, and composition residual
\[
C_{\mathrm{comp}}=\theta_{\mathrm{raw}}-\theta_{\mathrm{std}}.
\]
When $|C_{\mathrm{comp}}|+|\theta_{\mathrm{std}}|>0$, also report the bounded descriptive share
\[
\rho_{\mathrm{comp}}=
\frac{|C_{\mathrm{comp}}|}{|C_{\mathrm{comp}}|+|\theta_{\mathrm{std}}|}.
\]
No qualitative label is substituted for these quantities. The disposition is \textsc{estimated} or \textsc{blocked}.

\textbf{H5: No venue-independent dominance.} For each orientation profile in Section~\ref{subsec:pareto}, orient every coordinate so positive values favor Polymarket and let $\mathcal I_{p,d}$ be the compatible identified set for the venue difference. Polymarket is certified to dominate under profile $p$ only if $\inf\mathcal I_{p,d}\ge0$ for every load-bearing coordinate and $\inf\mathcal I_{p,d}>0$ for at least one; Kalshi dominance is defined symmetrically. H5 \textsc{passes} when neither venue is certified to dominate under every registered profile. It \textsc{fails--P} or \textsc{fails--K} when the same venue is certified to dominate under every profile. It is \textsc{blocked} when any profile lacks a load-bearing comparable coordinate.

\textbf{H6: Native-label schedule non-portability.} For each pre-registered label substitution $\ell$, let $A^{\mathrm{native}}_{\ell}(t)$ and $A^{\mathrm{sub}}_{\ell}(t)$ denote the native and naively substituted control states, and define
\[
\Delta T_{\ell}=\mu\{t:A^{\mathrm{native}}_{\ell}(t)\ne A^{\mathrm{sub}}_{\ell}(t)\}.
\]
H6 \textsc{passes} if at least one substitution has a lower identified bound exceeding the coarser timestamp precision of its two state histories. It \textsc{fails} if every evaluable substitution has $\Delta T_{\ell}=0$. It is \textsc{blocked} if no substitution has compatible state histories.

Every hypothesis can pass, fail, or remain blocked. Every registered question can be estimated or blocked. Neither is redefined on an opportunistic proxy.

\subsection{Planned comparative outputs}\label{subsec:planned-comparative-outputs}

When the comparison conditions are met, the study reports the population definition, semantic-match disposition, attrition, endpoint grades, exact paired contrasts, interval bounds, standardized unpaired contrasts, contestability and revision measures, observability and governance measures, capital-lock sensitivities, and the disposition of each registered question. Every reported table or figure identifies its denominator, statistical unit, endpoint map, evidence grade, calendar class, and reproducibility record.

\subsection{Failure and downgrade rules}\label{subsec:downgrade}

\begin{itemize}
  \item If Tier-A matching is sparse, report exact match scarcity and do not relax payoff semantics.
  \item If common-calendar path evidence is sparse, restrict paired results to terminal or interval-qualified coordinates and emphasize standardized unpaired distributions.
  \item If source/decidability clocks are sparse, remove those intervals rather than substituting close, request, or determination.
  \item If cash-realization evidence is incompatible, compare venue/protocol terminality and leave cash unranked.
  \item If common support fails, report native distributions and composition rather than extrapolating.
  \item If governance units cannot be aligned, retain separate native governance tables.
  \item If a reported analysis cannot be reproduced from its documented inputs, the dependent claim is not reported.
\end{itemize}

\begin{figure}[t]
\centering
\begin{tikzpicture}[node distance=4mm]
\node[pbox,text width=22mm] (u) {Venue-native authorities};
\node[pbox,text width=22mm,right=of u] (m) {Candidate and tier registry};
\node[pbox,text width=22mm,right=of m] (c) {Comparison certificates};
\node[pbox,text width=22mm,right=of c] (a) {Exact paired / interval paired / standardized unpaired};
\node[pbox,text width=22mm,right=of a] (g) {Evidence-qualified tables and figures};
\foreach \x/\y in {u/m,m/c,c/a,a/g}{\draw[parr] (\x)--(\y);}
\node[font=\scriptsize,align=center,below=8mm of c] {If required evidence is unavailable, the dependent result is suppressed or downgraded;\\the endpoint is not changed after inspection.};
\end{tikzpicture}
\caption{Analytical results pipeline. Reported outputs remain downstream of native evidence, semantic matching, and endpoint evidence.}
\label{fig:results-pipeline}
\end{figure}

\section{Implications for Leveraged Event Markets}\label{sec:leverage}

Paper 8 provides a comparison framework with venue-native empirical inputs; it is not a product specification. Its mechanism-design contribution is to identify which finality coordinates a leveraged overlay can treat as portable and which remain venue-specific.

\subsection{Obligation-specific terminal states}\label{subsec:obligation-terminal}

A financed obligation should terminate at the state that actually satisfies it. An oracle-final result may be sufficient for directional certainty but insufficient for claim delivery. A protocol-final claim may be redeemable while collateral has not been realized by a holder. A venue-final Kalshi position may be represented as settled without providing a public population clock for member withdrawability.

A cross-venue engine therefore needs an obligation map
\[
\mathcal O\mapsto k(\mathcal O),
\]
where $k(\mathcal O)$ is the functional finality coordinate that discharges obligation $\mathcal O$. The same synthetic product can use different native evidence adapters while preserving one functional discharge rule.

\subsection{State-triggered margin}\label{subsec:state-margin}

A margin schedule can condition on pre-final uncertainty, active contestability, final mechanism result, venue/protocol finality, or residual realization service. Paper 6 formalizes such controls \parencite{p6_2026}. Paper 8 adds a portability constraint: a schedule calibrated to a Polymarket dispute/reset state cannot be moved to a Kalshi label without proving that the target state covers the same residual risk and is observed with comparable delay.

\subsection{Finality carry}\label{subsec:finality-carry}

If ordinary price discovery or trading has ended but capital remains committed, any continuing charge should be linked to the service being provided. Before final mechanism result, it can compensate adjudication or review exposure. Between mechanism finality and venue/protocol terminality, it can compensate downstream execution or accounting. After technical finality, it can represent custody or realization service. Calling every interval ``funding'' hides the change in function.

\subsection{Conversion and delivery}\label{subsec:conversion}

Physical or claim conversion requires exact deliverable identity. On Polymarket this includes condition, outcome slots, payout vector, collateral route, and claim capacity. On Kalshi it would require a lawful and operational interface to the relevant venue position; Paper 8 does not assume one exists. Functional finality comparison can calibrate when a conversion would be safe, but it does not establish delivery authority.

\subsection{Cross-venue basis and latency arbitrage}\label{subsec:latency-arb}

Different finality paths can create intervals in which one venue has a first result or terminal state while the other remains contestable. A cross-venue trader can face basis risk, execution risk, or asymmetric close rather than a riskless arbitrage. Paper 8 specifies the matched-clock coordinates and endpoint alignment needed to measure such intervals once the semantic, calendar, and clock requirements are satisfied. A profitability claim additionally requires prices, fees, access, capital constraints, and executable depth, which are outside the finality comparison itself.

\subsection{Multi-leg instruments}\label{subsec:multi-leg}

For a spread containing one Polymarket and one Kalshi leg, each leg retains its native finality vector. The joint state becomes asynchronous when one leg reaches a functional coordinate before the other. Paper 9 formalizes the resulting residual exposure \parencite{p9_2026}. Paper 8 supplies endpoint alignment; it does not assume that equal native labels synchronize the legs.

\newpage
\section{Limitations}\label{sec:limitations}

The study has several limitations that remain substantive even after the empirical integration.

\textbf{Venue and version specificity.} The comparison concerns the pinned Polymarket and Kalshi architectures represented by the upstream authorities. API fields, rules, access, contracts, or review procedures can change. A later venue version requires a successor evidence snapshot.

\textbf{Asymmetric maturity.} The Polymarket evidence is published and historical. The Kalshi population and key-level finality follow-up are complete at their stated evidence grades, while contractual-decidability adjudication and cross-venue matching remain outstanding. Path, revision, and timer claims remain qualified where the public observation process is incomplete. The paper therefore combines two venue-native evidence layers without claiming a matched cross-venue timing result.

\textbf{Current-versus-historical field time.} The three collection identities and their series relations are independently validated as current observations after the enrollment freeze. Snapshot activity fields are also current observations. Neither layer is backprojected to historical rule state or activity time.

\textbf{Structural signatures.} Exact signature uniqueness is formula-specific. The registered signature excludes the output MVE ticker and preserves the exact selected-leg multiset, but includes event identity. It therefore measures exact tuple uniqueness, not independent information, liquidity, or source releases.

\textbf{Calendar mismatch.} The published Polymarket follow-up ends before the prospective Kalshi cohort begins. Exact event pairing must use an overlap-qualified historical tranche or a future concurrent population. Standardized unpaired comparison does not substitute for paired event evidence.

\textbf{Historical path loss.} The bounded Kalshi historical tranche can support terminal fields and semantic matching but cannot reveal transitions that were never preserved. A larger terminal archive would not solve that identification problem.

\textbf{Semantic adjudication.} Tier assignment includes human judgment about rule equivalence. Double review and documented evidence improve reproducibility but do not eliminate disagreement, especially for exceptional-event clauses and source corrections.

\textbf{Source and decidability missingness.} Polymarket does not have population-wide source-publication or contractual-decidability clocks in published evidence. Kalshi Paper 7.3 produces them only for an adjudicated subset. Comparisons on those coordinates can therefore be selected and small.

\textbf{Cash-finality asymmetry.} Polymarket exposes public redemption events, while comparable Kalshi member evidence can be authenticated and account-specific. Public venue finalization and holder redemption remain different endpoints. Cash comparisons can be unavailable even when venue-final comparisons are strong.

\textbf{Entitlement boundary.} Polymarket observed redemption does not identify the fraction of total winning entitlement redeemed. The missing balance-complete denominator prevents value-weighted completion comparisons.

\textbf{Internal governance.} Public event records do not disclose all internal review, escalation, or decision-making. Governance comparisons concern observable mechanisms and documented authority, not private organizational quality.

\textbf{Address identity.} Execution addresses are pseudonymous technical actors, not necessarily beneficial owners or independent decision makers. Address concentration is not directly comparable with venue-administered review actions.

\textbf{Selection and listing.} Events are not randomly assigned to venues. Paired availability, contract design, and listing choice can depend on event salience, source availability, and operational complexity. Standardization remains descriptive.

\textbf{Censoring and gaps.} Prospective capture can contain transport gaps; historical conditions can be left-truncated; late events can be administratively censored. Strictly qualified subsets can differ from the full population.

\textbf{Timestamp precision.} Exact timestamps can share one-second or block precision while secondary order differs or is unavailable. Zero duration does not imply simultaneity of internal processing, instant payout, or zero economic capital lock.

\textbf{No truth-accuracy ranking.} The paper does not determine whether a venue's final answer is substantively correct. Accuracy requires an external truth standard and treatment of ambiguity, revisions, and source integrity.

\textbf{No legal ranking.} Public technical finality is not a complete account of legal finality, customer property rights, insolvency treatment, or regulatory protection.

\textbf{No welfare function.} The multidimensional comparison does not choose universal weights for speed, contestability, observability, or capital lock. A normative ranking requires an explicit objective and stakeholder perspective.

\subsection{Kalshi-native evidence boundary}
The corrected Kalshi endpoint union is strong for venue-level finalization but asymmetric across surfaces. Most MVE endpoints are REST-only, while exact determination paths form a much smaller subset. Transport gaps do not erase exact two-clock durations, but they block population-wide claims about absence of intermediate revisions. Current settlement timers are not historically version-bound. The ordinary cohort is a reconstructed $T_0$ stock rather than a complete exact birth census. These qualifications apply to native Kalshi results before any cross-venue matching issue is considered.

\section{Conclusion}\label{sec:conclusion}
Polymarket and Kalshi now both have empirical native-finality authorities, but they still do not share one universal settlement clock. Polymarket separates proposal, dispute, Oracle finality, adapter terminality, protocol payout recording, redeemability, and observed holder redemption. Kalshi separates determination and venue finalization, while its public endpoint evidence is divided between lifecycle messages and exact REST settlement fields.

The Kalshi key-level follow-up shows why the observation process belongs inside the comparison. Exact public endpoints are identified for 7,357,576 of 7,611,594 MVE market objects, yet only 126,806 have an exact determination-to-endpoint pair. Most MVE exact endpoints are REST-only. In the reconstructed ordinary stock, 71,657 of 152,694 markets have an exact endpoint and 70,979 have an exact determination-to-endpoint pair. MVE exact pairs average 9.676 seconds and ordinary pairs 322.50 seconds, but these native distributions are not a causal architecture comparison and do not justify a venue winner.

Paper 8 therefore combines empirical native-finality evidence from both venues with a formal comparison framework. It aligns mechanism roles, grades observation quality, and defines standardized native distributions. A matched Polymarket--Kalshi contrast still requires blind contractual-decidability adjudication, payoff-functional matching, and common-calendar support for a shared event population.

The final result may be multidimensional. When neither venue dominates under the registered orientation profiles, the identified object is a trade-off surface rather than a scalar ranking. A venue can expose faster native endpoint clocks while offering different contestability, rule observability, holder-realization, or path-completeness properties. That is the appropriate input for event-linked leverage, margin, conversion, and capital-service design.

\appendix
\section{Functional Endpoint Dictionary}\label{app:endpoints}

\begin{longtable}{@{}>{\raggedright\arraybackslash}p{0.14\textwidth}>{\raggedright\arraybackslash}p{0.22\textwidth}>{\raggedright\arraybackslash}p{0.27\textwidth}>{\raggedright\arraybackslash}p{0.27\textwidth}@{}}
\caption{Operational endpoint dictionary.}\label{tab:endpoint-dictionary}\\
\toprule
Coordinate & Function & Polymarket candidate & Kalshi candidate \\
\midrule
\endfirsthead
\toprule
Coordinate & Function & Polymarket candidate & Kalshi candidate \\
\midrule
\endhead
$t^{\mathrm{src}}$ & Named public evidence release & Independently archived named source, generally unavailable at population scale & Exact source release object in Paper 7.3 subset \\
$t^{\mathrm{dec}}$ & Rules and available evidence imply one admissible result & Generally unmeasured in published Paper 5 evidence & First or stable contractual decidability in Paper 7.3 \\
$t^{\mathrm{close}}$ & Ordinary venue trading ceases for the contract & Market-specific close where independently available; not a Paper 5 finality state & Versioned market close \\
$t^{\mathrm{first}}$ & First recorded result under the native mechanism & First eligible Oracle proposal & First venue determination \\
$\mathcal I^{\mathrm{contest}}$ & One or more challenge/review episodes, their marks, and occupation & Proposal-liveness window(s), dispute, reset, or successor generation & Settlement-timer/review window(s), disputed/amended state, or re-determination; first activation reported separately when needed \\
$t^{\mathrm{mech}}$ & Final result of the adjudication mechanism & Successful final Oracle value & Last applicable determination after observed revisions \\
$t^{\mathrm{venue}}$ & Payout/accounting rule becomes terminal at venue or protocol & CTF payout recording; adapter terminal kept separately & Venue finalized/settled market state \\
$t^{\mathrm{cash}}$ & Holder/member realization & Observed redemption; positive endpoint separate & Member settlement evidence if comparable and lawfully observed \\
\bottomrule
\end{longtable}

\newpage
\subsection{Observation-code dictionary}\label{app:grade-dictionary}

\begin{table}[H]
\centering
\small
\begin{tabularx}{\textwidth}{@{}p{0.12\textwidth}Y@{}}
\toprule
Code & Meaning \\
\midrule
$E$ & Exact load-bearing native event or source timestamp with stable identity and documented precision. \\
$I$ & Transition lies between exact lower and upper observations. \\
$S$ & Current or terminal representation without transition history. \\
$Prx$ & Documented substitute with a stated functional imperfection. \\
$U$ & Available evidence does not identify the coordinate. \\
$X$ & Load-bearing evidence disagrees. \\
\bottomrule
\end{tabularx}
\end{table}

A coordinate can be exact while another coordinate for the same contract is unmeasured. Grades attach to endpoints, not entire markets.

\subsection{Contestability-code dictionary}\label{app:contestability-dictionary}

For every venue-native episode, the analytical relation stores: episode number; $t^{\mathrm{on}}$ lower/upper bound; $t^{\mathrm{off}}$ lower/upper bound; opening authority; closing authority; dispute/reset/amendment marks; result revision; censoring; and coordinate-specific grades. A scalar first-activation anchor is named explicitly as $t^{\mathrm{contest,on}}_1$. No generic $t^{\mathrm{contest}}$ is used in duration tables.

\section{Matched-Event Codebook}\label{app:codebook}

Each candidate pair stores the following fields before outcome comparison:
\begin{enumerate}
  \item exact venue contract identifiers and economic-event identifier;
  \item title and normalized proposition;
  \item polarity and payoff support;
  \item measured variable, unit, operator, threshold, and rounding;
  \item event interval, cut-off, and time zone;
  \item source identity, precedence, and release version;
  \item correction, revision, and preliminary/final policy;
  \item postponement, cancellation, invalidity, and void treatment;
  \item rule versions and retrieval/valid-time lineage;
  \item realized outcome and terminal payout;
  \item reviewer-one and reviewer-two field decisions;
  \item reconciled tier and disagreement record;
  \item calendar-support class;
  \item endpoint-map version and observation grades;
  \item economic-event/source-release inference cluster.
\end{enumerate}

\subsection{Multiplicity rules}\label{app:multiplicity}

When several venue contracts map to one economic event, the primary design either selects one pre-specified canonical contract per payoff functional or retains all child contracts with event-level clustering and explicit multiplicity weights. Selecting the fastest child after observing finality is prohibited. The comparison registry preserves rejected siblings and the selection rule.

\subsection{Tier-B sensitivity classes}\label{app:tier-b}

Tier B differences are coded by type rather than pooled:
\begin{itemize}
  \item source difference;
  \item threshold or rounding difference;
  \item timing-window difference;
  \item cancellation/void difference;
  \item correction-policy difference;
  \item payoff-support difference that permits a bounded scenario analysis;
  \item historical-rule incompleteness.
\end{itemize}
Each class has a pre-specified sensitivity interpretation. A Tier-B pair never silently enters the Tier-A denominator.

\section{Comparability Conditions and Claim Consequences}\label{app:gates}

\begin{longtable}{@{}>{\raggedright\arraybackslash}p{0.24\textwidth}>{\raggedright\arraybackslash}p{0.35\textwidth}>{\raggedright\arraybackslash}p{0.32\textwidth}@{}}
\caption{Cumulative comparability conditions for the empirical comparison.}\label{tab:gates}\\
\toprule
Condition & Evidence requirement & Claims unavailable when the condition is not met \\
\midrule
\endfirsthead
\toprule
Condition & Evidence requirement & Claims unavailable when the condition is not met \\
\midrule
\endhead
Polymarket input layer & Published Part I/II identities, states, denominators, and documented source records & Polymarket-native input claims \\
Kalshi registered population & Defined population boundary, native denominators, parent hierarchy, and cross-surface classification & Kalshi native population and observability claims \\
Current hierarchy fields & Documented current collection/series relation, signature formula, and activity-field dictionary & Current hierarchy, signature, and activity-field claims; not historical versions or trading time \\
Historical hierarchy fields & Historical collection/series/rule versions and time-indexed activity where claimed & Historical hierarchy, rule-state, or activity-path comparisons \\
Kalshi native lifecycle record & Qualified cohorts, endpoint-source union, exact clocks, and observation grades & Native endpoint/finality results; path, revision, and timer claims remain unavailable where their evidence grade is insufficient \\
Source and decidability record & Source object, rule graph, first/stable decidability, and documented review evidence & Source/decidability comparisons \\
Semantic equivalence & Independent payoff-functional review with retained tier & Membership in $\mathcal M_{\mathrm{sem}}$ \\
Exact paired observation & Tier A, common calendar support, exact compatible endpoint grades, and common censoring & Point-observed within-event contrast \\
Interval paired observation & Tier A, compatible exact/interval endpoint bounds, and common censoring & Bounded within-event contrast \\
Functional clock alignment & Equivalent endpoint function and sufficient observation grade & Endpoint-specific scalar duration or interval \\
Common risk set & Common origin and censoring rule & Survival or tail comparison \\
Common support & Declared support and stable standardization weights & Standardized unpaired contrast \\
Contestability evidence & Identified contestability episodes/occupation or explicitly graded marks & Cross-venue contestability duration or occupation comparison \\
Comparable governance unit & Aligned governance units or separately reported native metrics & Common governance-concentration comparison \\
Compatible cash evidence & Compatible holder/member realization evidence & Cash-finality comparison \\
Entitlement/value denominator & Exact entitlement or value denominator where claimed & Entitlement-completion or unredeemed-value comparison \\
Reproducible analysis record & Regeneration of each reported result from documented inputs & Release-bound empirical claims \\
\bottomrule
\end{longtable}

\subsection{Claim-strength labels}\label{app:claim-labels}

Every numerical statement carries one label:
\begin{itemize}
  \item \textsc{paired exact};
  \item \textsc{paired interval identified};
  \item \textsc{standardized unpaired};
  \item \textsc{native population exact};
  \item \textsc{audited subset exact};
  \item \textsc{snapshot descriptive};
  \item \textsc{proxy/model-based};
  \item \textsc{unmeasured};
  \item \textsc{blocked}.
\end{itemize}
A table can contain several labels when columns rely on different evidence classes.

\section{Statistical and Censoring Protocol}\label{app:statistics}

\subsection{Exact paired cohort}\label{app:exact-paired}

For each interval $(a,b)$, the exact paired cohort is $\mathcal M^{a,b}_{\mathrm{pair,exact}}$ and contains one comparison row per accepted pair and functional interval. Repeated child markets share an economic-event cluster. Descriptive summaries include the empirical CDF of within-pair contrasts, median and quantiles, sign shares, and registered tail thresholds. Any inferential interval states the finite-population or superpopulation interpretation explicitly.

\newpage
\subsection{Interval-qualified paired cohort}\label{app:interval-paired}

For each interval $(a,b)$, $\mathcal M^{a,b}_{\mathrm{pair,interval}}$ is reported as a denominator separate from the exact paired cohort. A pair enters this set only when Tier-A semantics and the interval comparison certificate pass; terminal-only semantic matches without two-sided bounds remain outside it.

\subsection{Interval bounds}\label{app:interval-bounds}

For interval contrasts $[L_i,U_i]$, the lower bound on the share with $\Delta_i\le c$ is the fraction with $U_i\le c$; the upper bound is the fraction with $L_i\le c$. Analogous bounds apply to positive/negative sign shares. The distribution of interval width is reported because wide bounds can dominate apparent uncertainty.

\subsection{Survival}\label{app:survival}

Venue-native survival curves use one functional origin and endpoint. Paired survival differences are reported only when calendar support and censoring align. Otherwise the paper reports separate native curves and numbers at risk. Markets entering after their origin are left-truncated; markets not reaching the endpoint by the fixed boundary are right-censored.

\subsection{Standardization}\label{app:standardization}

The reference population $w(x)$ is selected before outcome inspection. The primary candidates are: pooled overlap population, Tier-A candidate population, or a balanced target distribution over pre-specified strata. Positivity is checked by stratum counts and overlap diagnostics. Empty or near-empty cells are not filled by unconstrained extrapolation.

\subsection{Missingness diagnostics}\label{app:missingness}

For every endpoint, report observation-grade distribution by venue and covariates. A diagnostic model for exact-clock availability can be used to describe selection, but inverse-probability adjustment is secondary and requires stated assumptions. The primary result remains grade-stratified.

\subsection{Multiple testing}\label{app:multiple-testing}

The paper has a small number of registered outcome families. It emphasizes effect distributions and confidence intervals rather than a large collection of binary tests. Where several thresholds or dimensions are tested, family-wise or false-discovery adjustments are reported as sensitivity rather than used to hide the multivariate outcome structure.

\newpage
\section{Reproducibility}\label{app:reproducibility}

\subsection{Lineage}\label{app:lineage}

Every result follows the lineage
\begin{align*}
\text{venue-native evidence}
&\longrightarrow \text{functional endpoint map}
\longrightarrow \text{semantic pair},\\
\text{semantic pair}
&\longrightarrow \text{comparison certificate}
\longrightarrow \text{analytical row}
\longrightarrow \text{table or figure}.
\end{align*}
Each edge records the source identity, analytical definition, observation grade, and scope of the resulting claim.

\subsection{Immutable upstreams}\label{app:upstreams}

The published Paper 5 records are used as documented public inputs and are not recomputed opportunistically from current public sources. KMVE V1 defines the Kalshi population layer. The qualified Kalshi lifecycle evidence defines native cohorts, endpoint-source classes, exact clock pairs, and gap grades. A later API response cannot silently replace a missing value in the historical evidence record.

\subsection{Analytical reporting layer}\label{app:generated-layer}

Numerical tables and empirical figures are derived from structured analytical inputs. Each reported object identifies its denominator, endpoint map, evidence-grade filter, calendar class, and the definitions needed for independent reproduction.

\subsection{Reproducibility check}\label{app:offline-rebuild}

A reproducibility check verifies semantic relations, pair tiers, endpoint grades, calendar classes, censoring, analytical outputs, figure sources, bibliography, and the final PDF. Where file metadata can vary, the check compares definitions, key sets, semantics, and rendering rather than treating superficial byte variation as a scientific difference.

\begin{figure}[t]
\centering
\begin{tikzpicture}[node distance=5mm]
\node[pbox,text width=22mm] (raw) {Documented native evidence};
\node[pbox,text width=22mm,right=of raw] (native) {Native state histories};
\node[pbox,text width=22mm,right=of native] (map) {Functional endpoint maps};
\node[pbox,text width=22mm,right=of map] (pair) {Semantic pair and certificate};
\node[pbox,text width=22mm,right=of pair] (out) {Reported tables and figures};
\foreach \x/\y in {raw/native,native/map,map/pair,pair/out}{\draw[parr] (\x)--(\y);}
\node[psoft,text width=48mm,below=10mm of map] (rebuild) {Reproducibility check\\definitions, records, references, PDF render};
\draw[parr] (out.south) |- (rebuild.east);
\draw[parr] (raw.south) |- (rebuild.west);
\end{tikzpicture}
\caption{Reproducibility lineage. A current web or API response cannot silently replace a missing historical value in the documented evidence record.}
\label{fig:evidence-pipeline}
\end{figure}
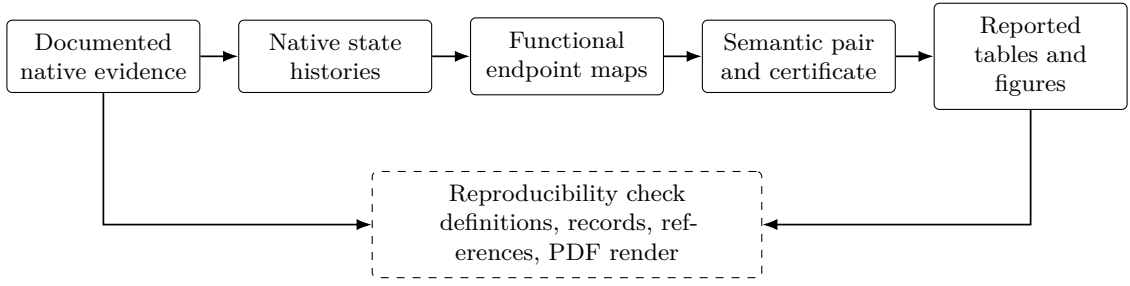

\newpage
\section{Historical-prefix validation of the Kalshi native layer}\label{sec:kalshi-historical-prefix}

The Kalshi native-finality inputs are estimated from the final registered enrollment and follow-up populations. To test the stability of the exact-clock machinery without substituting a different denominator, we separately reconstructed an earlier historical-prefix registry. The reconstruction preserves both published semantics---first determination to settlement and last determination to settlement---as separate records keyed by lifecycle stream, ticker, and pair definition.

For each semantic definition, the ordinary prefix contains 31,828 exact positive-duration pairs.  The MVE prefix contains 653,479 exact pairs: 130,973 positive-duration and 522,506 equal-second pairs at one-second timestamp precision.  Exact aggregate parity with the earlier audit is reproduced.  The bidirectional comparison with the final follow-up registry is nevertheless partial: 29,190 ordinary and 21,950 MVE pairs match exactly under each definition, while 2,638 ordinary and 631,529 MVE historical-prefix pairs do not match a final-registry pair and 41,789 ordinary and 104,856 MVE final-registry pairs are absent from the historical prefix.

\begin{table}[t]
\centering
\small
\caption{Kalshi historical-prefix registry and exact overlap with the final native follow-up registry. Counts are reported separately for each published pair definition; the two definitions have identical aggregate class counts but remain distinct registry records.}
\label{tab:historical-prefix-crosswalk}
\begin{tabularx}{\textwidth}{@{}lrrrr>{\raggedright\arraybackslash}X@{}}
\toprule
Population & Historical pairs & Exact overlap & Historical-only & Final-registry-only & Interpretation \\
\midrule
Ordinary & 31,828 & 29,190 & 2,638 & 41,789 & High historical-prefix overlap; the final registry additionally contains later or differently qualified pairs. \\
MVE & 653,479 & 21,950 & 631,529 & 104,856 & The historical prefix and final enrollment registry are substantially different observation populations. \\
\bottomrule
\end{tabularx}
\end{table}

The low MVE overlap is not treated as a venue-performance result.  It demonstrates that the historical prefix, the prospective REST enrollment population, and the completed WebSocket/REST follow-up registry are different observation surfaces.  The equal-second historical class likewise cannot be transported into a claim about payout, contestability, capital release, or Polymarket--Kalshi relative speed.  The comparative design therefore uses the final follow-up registry for Kalshi native endpoints and durations, and retains the historical-prefix reconstruction solely as a reproducibility check.  Matched cross-venue estimates remain conditional on semantic, endpoint, and calendar-support certificates.

\newpage
\section{Data, Disclosure, and Research Boundaries}\label{app:disclosures}

\subsection*{Data and Code Availability}
The Kalshi population and hierarchy layer is available as the Kalshi Multivariate Event Market Materialization Dataset (KMVE), Mendeley Data, V1, doi: 10.17632/fn65786cg6.1. The Polymarket upstream papers are public on SSRN. This paper reports venue-native inputs and no matched cross-venue numerical result.

\subsection*{Generative AI Disclosure}
OpenAI ChatGPT and Codex were used for editorial and technical assistance during manuscript preparation. The author made all substantive research decisions, reviewed the final manuscript, and assumes full responsibility for its contents.

\subsection*{Funding}
This research received no external funding.

\subsection*{Competing Interests}
The author is affiliated with the Research Department of Devnull FZCO and leads the ForesightFlow research programme. No external sponsor influenced the research design, analysis, interpretation, or decision to publish. The article does not evaluate a commercial product or make investment recommendations.

\printbibliography
\end{document}